\documentclass[11pt]{article}

\usepackage[a4paper,margin=1in]{geometry}
\usepackage[T1]{fontenc}
\usepackage[utf8]{inputenc}
\usepackage{lmodern}
\usepackage{microtype}
\usepackage{amsmath,amssymb,amsthm,mathtools}
\usepackage{booktabs}
\usepackage{array}
\usepackage{enumitem}
\usepackage{tikz}
\usepackage{float}
\usepackage{caption}
\usetikzlibrary{positioning}
\usepackage{hyperref}
\usepackage[nameinlink,noabbrev]{cleveref}

\hypersetup{
  colorlinks=true,
  linkcolor=black,
  citecolor=black,
  urlcolor=black
}

\newtheorem{definition}{Definition}[section]
\newtheorem{axiom}[definition]{Axiom}
\newtheorem{assumption}[definition]{Assumption}
\newtheorem{lemma}[definition]{Lemma}
\newtheorem{proposition}[definition]{Proposition}
\newtheorem{theorem}[definition]{Theorem}
\newtheorem{corollary}[definition]{Corollary}
\newtheorem{example}[definition]{Example}
\newtheorem{remark}[definition]{Remark}

\newcommand{\A}{\mathcal{A}}
\newcommand{\B}{\mathcal{B}}
\newcommand{\C}{\mathcal{C}}

\newcommand{\F}{\mathcal{F}}
\newcommand{\G}{\mathcal{G}}

\newcommand{\M}{\mathcal{M}}

\newcommand{\R}{\mathbb{R}}

\newcommand{\eps}{\varepsilon}
\newcommand{\TV}{\operatorname{TV}}
\newcommand{\Om}{\Omega}

\newcommand{\res}{\operatorname{res}}
\newcommand{\conv}{\operatorname{conv}}

\newcommand{\Ker}{\operatorname{Ker}}

\newcommand{\osc}{\operatorname{osc}}
\newcommand{\one}{\mathbf{1}}
\newcommand{\Ent}{\operatorname{H}}
\newcommand{\MI}{\operatorname{I}}

\title{Choric Masking in Ambient Release Systems:\\[3pt] {\large A Finite Certificate Calculus for Trace Indistinguishability under Bounded Audiences}}

\author{
Faruk Alpay\textsuperscript{*}\\
Department of Computer Engineering\\
Bah\c{c}e\c{s}ehir University, Istanbul, Turkiye\\
\texttt{faruk.alpay@bahcesehir.edu.tr}
\and
Taylan Alpay\\
Department of Aerospace\\
University of Turkish Aeronautical Association, Ankara, Turkiye\\
\texttt{s220112602@stu.thk.edu.tr}
}

\date{\vspace{-0.5em}\textsuperscript{*}Correspondence: \texttt{alpay@lightcap.ai}}

\begin{document}
\maketitle

\begin{abstract}
This paper develops a finite trace calculus for systems in which a protected coordinate is never observed directly yet remains statistically readable through visible roles, timing, repeated movement, bounded attention, and downstream object dynamics. An ambient release system is a staged probabilistic environment with hidden loci, norm gates, role masks, audiences, attention lenses, linked rooms, and post-release state. Its security notion, choric masking, asks not for invisibility but for non-singularity: the trace law induced by a protected locus must lie inside, or near, the convex hull of admissible cover traces, as measured by the tests actually available to an audience. Every certificate in the calculus is finite-dimensional. Trace laws form polytopes, audiences induce measurement operators, masking is intersection in the projected measurement space, and exposure is witnessed by separating hyperplanes, kernel obstructions, hypothesis-testing bounds, and Fano-type information lower bounds. We prove measurement-polytope equivalence for exact and approximate masks, dual separation certificates, localization conditions that distinguish group residue from carrier residue, data-processing laws for bounded attention, aperture identities for gaze-thinned observation, lower bounds for mandatory unique gestures, composition rules for linked releases, and support-separation certificates for post-release exposure. A repeated-room theory defines choric risk debt and proves that unresolved pressure can broaden selection and shift cost onto cover populations without increasing carrier hazard, separating absent, unseen, displaced, and delayed residue. The calculus gives designers and auditors of public-facing systems finite, checkable certificates for each of these distinctions.
\end{abstract}

\medskip
\noindent\textbf{Keywords:} trace indistinguishability; anonymity and unlinkability; side channels; leakage certificates; bounded observers; privacy; surveillance and inspection systems; accountability; institutional design.

\section{Introduction}

Many public systems try to protect a hidden coordinate while still requiring a visible performance. A person, institution, device, or record must pass through a sequence of ordinary actions: enter, wait, declare, correct, acknowledge, defer, forward, discard, continue, or leave. The protected coordinate is not directly printed on the public surface, yet the surface is not neutral. Timing, repetition, hesitation, routing, role mismatch, object persistence, and audience attention can create a statistical trace. The problem studied here is the gap between formal non-observation and practical readability.

The paper treats that gap as a finite trace problem. A hidden locus chooses or induces visible movement through a norm gate; an audience sees only an observation sequence; and the security question is whether two hidden loci induce trace laws that are indistinguishable for the tests actually available to the audience. This is close in spirit to simulation and leakage definitions~\cite{goldwasser1989,canetti2001,lindell2017,dwork2006}, but the channel of concern is neither a message transcript alone nor a database answer alone. The channel is the public choreography around the release: who moves, when they move, what they carry, which ordinary roles surround them, how attention is allocated, and what later rooms force them to do.

The central term is \emph{choric masking}. A locus is chorically masked when its trace can be absorbed into a chorus of admissible traces generated by other loci, roles, objects, or cohorts. Masking is therefore not invisibility. It is non-singularity under a specified audience class. A public residue may be real while its carrier remains non-localizable; conversely, a carrier may pass one gate and later become exposed because a carried state changes the set of admissible future scripts.

The manuscript is organized around four claims, each made as a theorem rather than as metaphor.
\begin{enumerate}[label=(\arabic*),leftmargin=2.2em]
\item \textbf{Trace-level masking.} For finite horizons, masking is equivalent to approximate intersection of convex hulls of admissible trace laws. The witness is a coupling or a common point in a measurement polytope.
\item \textbf{Attention-level masking.} A bounded observer does not test the full trace. A stochastic attention lens first thins the trace. Data processing then implies that attention can only reduce the distinguishability available to the observer, and aperture bounds quantify when a real residue becomes practically unreadable.
\item \textbf{Post-release exposure.} Some hidden coordinates are not inert. A carried object or state may decay, become costly, become scarce, create hazard, or demand storage. This residual pressure narrows later scripts and can create a second-room exposure floor even after first-room masking succeeds.
\item \textbf{Repeated unresolved pressure.} When a room reacts to unresolved residue without localizing the source, risk becomes a debt variable. The response may broaden, false positives may rise, per-candidate attention may fall, and the original carrier hazard need not increase. This separates cumulative system pressure from personal non-capture.
\end{enumerate}

The calculus is written for the designer and the auditor, and its outputs are certificates: a separating test, a hull intersection, a coupling witness, an information bound, a pressure ledger. Each certificate states which of three events a public system is currently conflating: absence of residue, residue outside the audience aperture, and residue displaced onto a cover population. The design implication is exact: broad response is not a substitute for localization, and a system that cannot attach pressure to the layer that generated it will charge that pressure to unrelated participants.

\subsection*{Threat model and scope}

A hidden locus is any protected coordinate whose public readability is at issue: a role, object class, affiliation, source, condition, or local state. An audience is not assumed to be omniscient. It has a test class, a finite attention budget, and possibly an adaptive response policy. The model studies statistical separation of trace laws under those constraints. It does not assume malice, guilt, or innocence of any particular carrier; those words are outside the mathematics. The object of analysis is whether the room can localize the source of a residue without charging the wrong layer.

Two audience grades recur. A \emph{passive} audience applies a fixed finite family of bounded tests to the public trace. An \emph{adaptive} audience also carries a response policy whose state (attention allocation, selection breadth, accumulated debt) evolves with unresolved residue; its power is limited by explicit information and budget bounds rather than by fiat. All statements concern distributions over traces and are invariant under relabeling of carriers inside a cover cell, which is exactly the quantity the localization certificates measure. This fixes the intended use of the theory: the calculus certifies when a room can attribute a residue and when it cannot, so that pressure is either localized or recognized as displaced.

All spaces are finite unless explicitly stated otherwise. This is deliberate. Finiteness makes every certificate constructive: residues have cylinder witnesses, masks have finite convex certificates, and failure modes can be rendered as inequalities over budgets, hulls, links, and pressures. The resulting statements are not asymptotic slogans. They are finite horizon claims about what a bounded audience can and cannot infer from visible traces.

\section{Relation to existing frameworks}
\label{sec:related}

The calculus sits at the junction of four literatures, and its contribution can be stated exactly against each.

\paragraph{Cryptographic indistinguishability.}
Simulation-based definitions compare a real transcript with an ideal one and certify protocols by reductions~\cite{goldwasser1989,bellare1998,canetti2001,lindell2017}. The present object differs in channel and in proof move. The channel is not a protocol transcript but the public choreography around a release: who moves, when, in which role, under which gate. The proof move is not a reduction but a finite separating certificate over convex sets of admissible behaviors. Masking is a property of a population of trace laws, so no single participant is required to imitate anyone; the room as a whole must supply a hull.

\paragraph{Statistical disclosure and differential privacy.}
Differential privacy constrains how strongly a curated mechanism may depend on one record~\cite{dwork2006,dworkroth2014,nissim2017}, and its operational content is a bound on binary hypothesis tests~\cite{wasserman2010}. Ambient release has no curator: the release is the behavior itself, produced under a norm gate, and the protected coordinate enters through transition and emission kernels rather than through a database row. The hypothesis-testing reading survives the change of setting --- $\eps$-masking bounds the advantage of every audience-implementable test (\Cref{lem:ht}) --- but the protected object becomes membership of an induced law in a cover hull, and the attack surface includes trace-level carriers of the kind documented in de-anonymization and traceability results~\cite{narayanan2008,dwork2015}. The two frameworks compose rather than compete: a privately curated release inside a room is one admissible script among others, and the room may still leak through timing, roles, and burden.

\paragraph{Anonymity systems and their metrics.}
Mixes and crowds protect a sender by making them one of many plausible sources~\cite{chaum1981,reiter1998}, with protection quantified by anonymity-set size or entropy~\cite{serjantov2002,diaz2002}, standardized as unlinkability~\cite{pfitzmann2010}, and ported to records as $k$-anonymity~\cite{sweeney2002}. Choric masking generalizes the protected population from users behind a relay to layer values behind a measurement hull: cover multiplicity counts the values whose hulls contain the observed measurement, the Fano bound converts probable innocence into a quantitative localization floor parameterized by the information actually carried to the response policy, and the Helly certificate bounds the number of comparisons needed to refute a group cover. The framework also keeps what relay-based metrics drop: gates that force gestures, observers with finite apertures, and releases whose consequences arrive later.

\paragraph{Quantitative information flow.}
Quantitative information flow measures leakage of a secret through a channel with gain functions and operational leakage measures~\cite{smith2009,alvim2020,issa2020}, and its data-processing inequality drives every post-processing argument. The attention lens of this paper is channel pre-composition read through that inequality, and the kernel-enlargement theorem is its nullspace form. Four objects here lie outside the standard inventory: admissible-script polytopes, which make the secret's feasible behavior set itself a convex body shaped by a norm gate; separating-hyperplane certificates computable by finite linear programs over those polytopes; post-release burden, a support-level rather than likelihood-level channel by which a carried state narrows future scripts; and debt dynamics, which track where pressure lands when localization fails.

\paragraph{Social accounts of public visibility.}
Dramaturgical role performance~\cite{goffman1959}, contextual integrity~\cite{nissenbaum2004}, networked privacy~\cite{boyd2012}, obfuscation as deliberate cover traffic~\cite{brunton2015}, and analyses of categorical screening and its distributional costs~\cite{gandy2006,tufekci2014,kearns2019} describe the same room from the outside. The displaced-risk theorems give these accounts a checkable form: cost shifted onto a cover population is here a certified inequality between a debt trajectory, a dilution bound, and a spillover derivative, not a metaphor.

To our knowledge, the combination is new: hull-intersection masking under audience seminorms, constructive separation certificates, attention-thinned readings, burden-induced second-room floors, and a debt calculus for repeated unlocalized response, all inside one finite and inspectable model.

\section{Proof architecture}

The proof system is now organized around a single finite-dimensional object: the signed difference of two public trace laws.  Every later notion is a way of applying a linear or stochastic map to that difference and then asking whether any admissible functional still separates it from zero.  This prevents the manuscript from relying on metaphor.  A ``room'' is a finite carrier, a ``chorus'' is a convex hull of trace laws, a ``residue'' is a signed measure, an ``audience'' is a family of bounded linear readings, an ``attention lens'' is a Markov kernel, a ``linked room'' is another Markov kernel, and ``pressure'' is a state-dependent restriction of the downstream support.

The dependency chain is as follows.  First, finite scripts generate trace polytopes.  Second, an audience induces a measurement operator on those polytopes.  Choric masking is exactly intersection, or approximate intersection, after this operator is applied.  Third, if the projected polytopes do not intersect, finite-dimensional separation produces a signed audience certificate; this is the real proof object.  Fourth, localization is not assumed from separation: it requires a second certificate showing that the separating functional attaches to a layer rather than merely to a group.  Fifth, an attention lens composes with the measurement operator, so its effect is the replacement of one kernel by a larger one; any residue in that new kernel is unobservable to that audience.  Sixth, post-release pressure is modeled by nested downstream script sets.  Support separation in those nested sets creates a later exposure floor even when the entrance projection intersects.  Finally, repeated rooms add a reflected debt recursion.  Debt is not hazard unless another certificate converts aggregate pressure into localized attention, localized residue, or information about the carrier.

This order matters.  The paper does not prove that a vague environment ``feels risky.''  It proves statements of the following form: a signed law lies outside a kernel, a projected polytope is separated by a hyperplane, a Markov lens contracts the relevant seminorm, a support inclusion is strict, or a debt recursion grows without a matching localization certificate.  Those are the only admissible proof moves.

\section{Model overview: staged rooms and visible traces}

A public system often gives its participants a narrow set of movements. One may enter, wait, acknowledge, correct, refuse, forward, remain silent, or leave. Each movement is ordinary when seen alone. A problem appears only when a hidden locus is read through the sequence. The audience may not observe an internal state, a private motive, or a protected attribute, but it may observe a rhythm of hesitation, a pattern of correction, a missing reply, or a change in the crowd around the participant. The leakage is not a single secret bit written on a wall. It is a residue left by movement inside a room.

The aim of this paper is to give that room a precise mathematical form. We study finite staged environments in which a hidden locus produces a public trace through a visible role. The audience is modeled by a family of bounded tests on traces. The room is secure for a pair of hidden loci when no permitted audience test separates the two trace laws by more than a prescribed tolerance. The phrase \emph{choric masking} is used because the trace of one locus is protected only when it can be placed inside a chorus of other admissible traces. The participant is not erased. The participant is made non-singular.

The paper deliberately avoids a story in which protection is reduced to encryption, access control, or a private mechanism. Those objects may be present in an implementation, but they are not the object studied here. The primitive object is the visible trace. This choice matters because many release systems fail after the protected channel has already done its job. A message may be hidden, but the timing of its release, the shape of its acknowledgments, or the roles forced by the surrounding institution may still expose the source.

The room has six components.
\begin{enumerate}[label=(\roman*),leftmargin=2.2em]
\item A finite set of hidden loci, denoted by $\Theta$.
\item A finite set of visible roles, denoted by $R$.
\item A finite state space $S$ and a finite action alphabet $A$.
\item A finite observation alphabet $Y$ from which public traces are built.
\item A norm gate that decides which actions are admissible for a role at a visible history.
\item A class of audience tests that assigns a number in $[0,1]$ to a trace.
\end{enumerate}
The hidden locus affects the evolution of the room, but the audience sees only $Y$-valued traces. A role is a public costume. A locus is the protected coordinate behind the costume. A norm gate is the social and institutional grammar of the room.

The main results are as follows.
\begin{itemize}[leftmargin=2em]
\item A coupling criterion converts indistinguishability into the construction of a joint staged process whose observable statistics agree with high probability.
\item A convexity theorem shows that choric masks are exactly common points, or approximate common points, in convex hulls of admissible trace laws.
\item A stratified localization calculus separates person, object, role, companion, cohort, clock, and side-release layers; exposure becomes a statement about which layer an audience can actually localize.
\item A residue theorem says that whenever an audience can separate two loci, the separation has a finite cylinder witness. Leakage is therefore not mystical: it is carried by a concrete pattern of visible gestures.
\item A cover-multiplicity theorem proves that a real residue may still fail to identify its carrier when enough ordinary trace sources share the same measurement hull.
\item An attention-lens theorem proves that a residue visible in the full trace may vanish for a bounded observer whose gaze, dwell time, and working aperture thin the trace before testing.
\item A unique-gesture lower bound proves that a gesture used by only one hidden layer-value cannot be masked unless the room supplies enough cover mass or the diagnostic coordinates fall outside the effective aperture.
\item A repeated-window theorem shows how a tiny residue becomes large when the room is visited many times.
\item A choric-debt theorem shows how unresolved residue can accumulate into global pressure, broaden selection, dilute per-candidate attention, and externalize cost to cover populations without resolving the original carrier.
\item A composition theorem gives a leakage accounting rule for linked rooms and delayed release gates.
\item A residual-pressure theorem proves that a carried object can create a second-room exposure even after first-room masking has succeeded.
\item A burden certificate separates three failures that are often conflated: no residue, residue without attention, and residue whose downstream pressure is not yet visible.
\end{itemize}
The later sections turn these results into a worked finite example. The example is not empirical. It is a laboratory in which the reader can compute exactly where the exposure sits.

\section{Traces, tests, and ambient distance}

Let $Y$ be a finite observation alphabet. For a horizon $T\geq 1$, write $Y^T$ for the set of public traces $y_{1:T}=(y_1,\ldots,y_T)$. A probability law $P$ on $Y^T$ is a full description of what the audience sees over $T$ visible steps.

\begin{definition}[Audience test]
An audience test over horizon $T$ is a function
\[
    f:Y^T\to [0,1].
\]
A test class is a set $\F_T\subseteq [0,1]^{Y^T}$. It is symmetric if $f\in\F_T$ implies $1-f\in\F_T$.
\end{definition}

The value $f(y_{1:T})$ may be read as suspicion, confidence, attention, severity, or selection probability. No interpretation is fixed. The mathematics only requires boundedness.

\begin{definition}[Ambient distance]
For probability laws $P,Q$ on $Y^T$ and a test class $\F_T$, define
\[
    \Delta_{\F_T}(P,Q)
    := \sup_{f\in\F_T}\left| \mathbb E_P f-\mathbb E_Q f\right|.
\]
When $\F_T=[0,1]^{Y^T}$, this is the total variation distance $\TV(P,Q)$.
\end{definition}

The distance is not a property of the trace laws alone. It is a property of the trace laws under a chosen audience. If the audience can count exact timestamps, the distance may be large. If the audience only sees coarse phases, the same two laws may be close. This dependence is essential. A public system is not experienced by an all-seeing statistician. It is experienced through the tests available to its audience.

\begin{lemma}[Operational reading: hypothesis testing]\label{lem:ht}
Let $P,Q$ be laws on $Y^T$ and let $\F_T$ be a test class. Interpret $f\in\F_T$ as a randomized accept rule for the hypothesis that the trace was produced under $P$. Then for every $f\in\F_T$,
\[
    \underbrace{\mathbb E_Q f}_{\text{false acceptance}}
    \;+\;
    \underbrace{\mathbb E_P (1-f)}_{\text{false rejection}}
    \;\geq\; 1-\Delta_{\F_T}(P,Q).
\]
Thus $(\F_T,\eps)$-closeness of two trace laws states that no audience-implementable decision rule distinguishes them with total error below $1-\eps$.
\end{lemma}

\begin{proof}
By definition $\mathbb E_P f-\mathbb E_Q f\leq \Delta_{\F_T}(P,Q)$. Rearranging gives $\mathbb E_Q f+1-\mathbb E_P f\geq 1-\Delta_{\F_T}(P,Q)$.
\end{proof}

\begin{remark}
This is the test-class restriction of the binary hypothesis-testing identity that underlies statistical readings of privacy definitions~\cite{wasserman2010}. The restriction is the point: the room is experienced through $\F_T$, not through the full power set.
\end{remark}

\begin{definition}[Cylinder]
For $I\subseteq \{1,\ldots,T\}$ and $u\in Y^I$, the cylinder determined by $(I,u)$ is
\[
    C(I,u)=\{y_{1:T}\in Y^T: y_i=u_i\text{ for every } i\in I\}.
\]
The full cylinders are the singleton cylinders with $I=\{1,\ldots,T\}$.
\end{definition}

\begin{definition}[Residue]
For laws $P,Q$ on $Y^T$ and an event $C\subseteq Y^T$, the signed residue of $C$ is
\[
    \res_{P,Q}(C)=P(C)-Q(C).
\]
Its magnitude $|\res_{P,Q}(C)|$ is the exposure carried by $C$.
\end{definition}

The word \emph{residue} is useful because many leaks are not direct labels. They are leftovers after the expected public story is subtracted from the observed one. A rare correction, a silence at the wrong tick, or a delayed handoff may carry a residue even if it is a valid move.

\begin{lemma}[Elementary properties]
For every test class $\F_T$ and laws $P,Q,H$ on $Y^T$:
\begin{enumerate}[label=(\alph*),leftmargin=2em]
\item $0\leq \Delta_{\F_T}(P,Q)\leq 1$.
\item $\Delta_{\F_T}(P,Q)=\Delta_{\F_T}(Q,P)$.
\item $\Delta_{\F_T}(P,H)\leq \Delta_{\F_T}(P,Q)+\Delta_{\F_T}(Q,H)$.
\item If $\F_T\subseteq \G_T$, then $\Delta_{\F_T}(P,Q)\leq \Delta_{\G_T}(P,Q)$.
\end{enumerate}
\end{lemma}

\begin{proof}
Boundedness gives (a), symmetry of the absolute value gives (b), and the triangle inequality for real numbers gives (c). Inclusion of test classes gives (d).\end{proof}

\section{Ambient release systems}

We now describe the staged room. The definition is finite because finiteness forces all claims to be inspectable. There are no measurability exceptions hidden in the walls.

\begin{definition}[Ambient release system]
An ambient release system is a tuple
\[
    \A=(\Theta,R,S,A,Y,s_0,K,O,N,\rho),
\]
where:
\begin{enumerate}[label=(\roman*),leftmargin=2.2em]
\item $\Theta$ is a finite set of hidden loci.
\item $R$ is a finite set of visible roles.
\item $S$ is a finite state space with initial state $s_0\in S$.
\item $A$ is a finite action alphabet.
\item $Y$ is a finite observation alphabet.
\item $K:S\times A\times \Theta\to \Delta(S)$ is a transition kernel.
\item $O:S\times A\times R\times \Theta\to \Delta(Y)$ is an observation kernel.
\item $N\subseteq R\times Y^{<T}\times A$ is the norm gate. The statement $(r,h,a)\in N$ means that role $r$ may use action $a$ after visible history $h$.
\item $\rho\in\Delta(R)$ is a public role prior, used when the room assigns roles without revealing the hidden locus.
\end{enumerate}
Here $\Delta(E)$ denotes the set of probability distributions on a finite set $E$.
\end{definition}

The hidden locus is allowed to affect both state transition and emission. This is not because the locus is morally responsible for the trace. It is because the trace may be statistically shaped by conditions that are invisible to the audience. For instance, one locus may face a delay, another may face a stricter gate, and a third may have fewer safe acknowledgments available.

\begin{definition}[Admissible script]
Fix $T$. A script for role $r$ is a sequence of stochastic kernels
\[
    \pi_t(\cdot\mid h)\in \Delta(A),\qquad h\in Y^{t-1},\quad 1\leq t\leq T.
\]
It is $N$-admissible if
\[
    \pi_t(a\mid h)>0 \quad\Longrightarrow\quad (r,h,a)\in N
\]
for every $t,h,a$. The set of all $N$-admissible scripts for role $r$ is denoted by $\Pi_T(r)$.
\end{definition}

A script is not a hidden algorithm in the sense of a protected implementation. It is a public-form movement rule constrained by the room. It may be randomized, but only among actions allowed by the norm gate.

\begin{definition}[Trace law]
For hidden locus $\theta\in\Theta$, role $r\in R$, and script $\pi\in\Pi_T(r)$, the induced trace law $P_{\theta,r}^{\pi,T}$ on $Y^T$ is the marginal law of the process
\[
    S_0=s_0,
\]
then, for $1\leq t\leq T$,
\[
    A_t\sim \pi_t(\cdot\mid Y_{1:t-1}),\qquad
    Y_t\sim O(\cdot\mid S_{t-1},A_t,r,\theta),\qquad
    S_t\sim K(\cdot\mid S_{t-1},A_t,\theta).
\]
\end{definition}

The audience does not see $S_t$, $A_t$, $r$, or $\theta$ unless these are encoded in $Y_t$. The trace law is therefore the public shadow of the staged process.

\begin{definition}[Admissible trace set]
The admissible trace set of hidden locus $\theta$ at horizon $T$ is
\[
    \M_T(\theta)=\left\{P_{\theta,r}^{\pi,T}: r\in R,\ \pi\in\Pi_T(r)\right\}.
\]
Its choric hull is
\[
    \operatorname{Ch}_T(\theta)=\conv \M_T(\theta).
\]
A point in $\operatorname{Ch}_T(\theta)$ is a mixture of admissible role-script traces for locus $\theta$.
\end{definition}

The hull represents the chorus. A system may not be able to make every individual trace harmless, but it may be able to arrange a population of traces so that the public distribution does not single out a locus.

\section{The axiomatic floor}

The definitions above are not merely descriptive. They impose a small set of axioms on what counts as a room. We record them explicitly because every later certificate uses them. The point of the axioms is to prevent drift: a protected locus may be discussed only through visible traces, admissible continuations, mixtures of public scripts, and tests available to an audience.

\begin{axiom}[Finite public carrier]
For each horizon $T$, the public carrier is the finite set
\[
    \Om_T=Y^T.
\]
Every security statement at horizon $T$ is a statement about probability laws in $\Delta(\Om_T)$.
\end{axiom}

\begin{axiom}[Gate legality]
A script may place positive probability only on actions admitted by the norm gate. Equivalently,
\[
    \pi_t(a\mid h)>0 \quad\Longrightarrow\quad (r,h,a)\in N.
\]
No later mixture, quotient, or side channel may reintroduce mass on an action that the gate forbids at the moment of choice.
\end{axiom}

\begin{axiom}[External choric mixing]
If a public population can realize trace laws $P_1,\ldots,P_m$ for the same hidden locus, then it can realize every externally mixed law
\[
    \sum_{i=1}^m \alpha_i P_i,
    \qquad \alpha_i\geq 0,
    \qquad \sum_i\alpha_i=1.
\]
The mixing variable is not a new secret. It is a population-level or ceremony-level randomizer outside the single trace.
\end{axiom}

\begin{axiom}[Audience boundedness]
An audience at horizon $T$ is represented by bounded tests $f:\Om_T\to[0,1]$. If an audience can use a test $f$, then it can use its refusal $1-f$. Passing to convex mixtures of tests does not create a stronger audience.
\end{axiom}

\begin{axiom}[Parastatic naturality]
Any auxiliary release, delay, label, moderator note, or downstream republication is represented by a Markov link
\[
    L:\Om_T\to\Delta(\Xi)
\]
for some finite public side carrier $\Xi$. A link that is common to two loci is not allowed to increase their full-observation distance.
\end{axiom}

The axioms make the room finite-dimensional. This is not cosmetic. It means that exposure is not hidden in an informal narrative; it is a vector, a hyperplane, or a finite residue.

\begin{definition}[Trace polytope]
For a hidden locus $\theta$, the trace polytope at horizon $T$ is
\[
    C_T(\theta)=\operatorname{Ch}_T(\theta)\subseteq\Delta(\Om_T).
\]
Its extreme generators are the trace laws induced by deterministic admissible scripts.
\end{definition}

\begin{proposition}[Polytope consequence]
For every hidden locus $\theta$ and horizon $T$, $C_T(\theta)$ is a nonempty compact convex polytope. In particular, all infima and suprema in the definitions of choric distance, exposure radius, and room radius are attained.
\end{proposition}

\begin{proof}
The history tree $Y^{<T}$ is finite and the action alphabet is finite. A deterministic admissible script chooses one legal action at each reachable history, so there are only finitely many such scripts. A stochastic admissible script is a product of local distributions over finite legal action sets. Sampling all local choices in advance expands it into a mixture of deterministic admissible scripts. Hence the induced trace laws of stochastic scripts lie in the convex hull of finitely many deterministic trace laws. This hull is a compact convex polytope in the finite-dimensional simplex $\Delta(\Om_T)$. Nonemptiness follows whenever every reachable gate has at least one legal action; otherwise the room has a dead history and is excluded at that horizon. Continuity of $\Delta_{\F_T}$ on a compact finite-dimensional set gives attainment.
\end{proof}

\begin{definition}[Legibility functional]
A test $f:\Om_T\to[0,1]$ induces the signed legibility functional
\[
    \ell_f(P,Q)=\mathbb E_P f-\mathbb E_Q f.
\]
It is a public hyperplane reading one law against another.
\end{definition}

\begin{lemma}[Convex closure of the audience]
Let $\operatorname{co}(\F_T)$ be the convex hull of $\F_T$. Then
\[
    \Delta_{\operatorname{co}(\F_T)}(P,Q)=\Delta_{\F_T}(P,Q).
\]
If $\F_T$ is closed under complements, then
\[
    \Delta_{\F_T}(P,Q)=\sup_{f\in\F_T}\big(\mathbb E_P f-\mathbb E_Q f\big).
\]
\end{lemma}

\begin{proof}
The map $f\mapsto \mathbb E_P f-\mathbb E_Q f$ is affine. The supremum of an affine functional over a convex hull is already attained on the original set of generators. The absolute value is handled by the complement rule: replacing $f$ by $1-f$ changes the sign of the difference.
\end{proof}

\begin{theorem}[Dual residue certificate]
Let $C,D\subseteq\Delta(\Om_T)$ be compact convex trace polytopes and define
\[
    d_{\rm TV}(C,D)=\inf_{P\in C,Q\in D}\TV(P,Q).
\]
If $d_{\rm TV}(C,D)>0$, then there exists a bounded public score $g:\Om_T\to\R$ with oscillation
\[
    \osc(g)=\max_{\omega}g(\omega)-\min_{\omega}g(\omega)\leq 1
\]
such that, after orienting $C$ and $D$,
\[
    \inf_{P\in C}\mathbb E_P g-
    \sup_{Q\in D}\mathbb E_Q g
    \geq d_{\rm TV}(C,D).
\]
For any closest pair $(P^\star,Q^\star)$, the event
\[
    E^\star=\{\omega:P^\star(\omega)>Q^\star(\omega)\}
\]
attains the pairwise gap
\[
    P^\star(E^\star)-Q^\star(E^\star)=\TV(P^\star,Q^\star).
\]
The first object is a uniform separating certificate for the two polytopes; the second object is the concrete residue carrier for the closest pair.  They need not be the same object.
\end{theorem}

\begin{proof}
Work in the vector space $V=\{\mu\in\R^{\Om_T}:\sum_\omega\mu(\omega)=0\}$ of signed mass-zero measures and write $H=C-D=\{P-Q:P\in C,Q\in D\}$.  The set $H$ is compact and convex.  The assumption $d_{\rm TV}(C,D)>0$ means that $0\notin H$ and that the closed TV ball of radius $d_{\rm TV}(C,D)$ around zero touches $H$ but does not enter its interior.  Finite-dimensional convex separation gives a linear functional $g$ separating $H$ from the open TV ball with dual norm at most one.  The dual unit ball of TV on $V$ is the set of functions with oscillation at most one, because adding a constant to $g$ does not change $\langle g,\mu\rangle$ for $\mu\in V$.  Hence
\[
    \langle g,P-Q\rangle\geq d_{\rm TV}(C,D)
    \qquad(P\in C,Q\in D),
\]
which is the displayed uniform margin.  For the closest pair, the usual finite formula for total variation gives
\[
    \TV(P^\star,Q^\star)=\sum_{\omega:(P^\star-Q^\star)(\omega)>0}
        (P^\star(\omega)-Q^\star(\omega)),
\]
which is exactly the event gap on $E^\star$.
\end{proof}

\begin{corollary}[No invisible positive distance]
If two loci have positive full-audience choric distance, then there is a finite public event whose probability gap is positive for every pair of optimally chosen choric laws. Thus every positive exposure floor has a public carrier.
\end{corollary}

\begin{proof}
Apply the dual residue certificate to $C_T(\theta)$ and $C_T(\theta')$.
\end{proof}

\begin{theorem}[Radius as a minimax program]
For a finite event audience generated by $G_1,\ldots,G_k\subseteq\Om_T$, the exposure radius of $\theta$ against $B$ is the optimum of the finite program
\[
\min_{P\in C_T(\theta)}\ \max_{\theta'\in B}\ \min_{Q\in C_T(\theta')}\ \max_{1\leq i\leq k}
    |P(G_i)-Q(G_i)|.
\]
If each polytope is written by deterministic-script generators, this program is a nested finite-dimensional linear optimization problem; fixing the outer comparison locus makes it a linear program after epigraph variables are introduced.
\end{theorem}

\begin{proof}
For an event audience, $\Delta_{\F_T}$ is exactly the maximum absolute difference of the listed event masses. Each event mass is a linear function of the mixture weights over deterministic-script generators. The maximum absolute value is linearized by a scalar epigraph variable $t$ with constraints
\[
    -t\leq P(G_i)-Q(G_i)\leq t,
    \qquad i=1,\ldots,k.
\]
The remaining statements follow from compactness of the trace polytopes and the preceding polytope consequence.
\end{proof}

\section{Canonical certificate calculus}

The preceding axioms become useful only after they are turned into a certificate calculus.  This section supplies that calculus.  It is the mathematical core of the manuscript: masking, attention loss, localization failure, and delayed pressure are all expressed as statements about affine images of finite polytopes.

Let
\[
    V_T=\left\{\mu\in\R^{\Om_T}:\sum_{\omega\in\Om_T}\mu(\omega)=0\right\}
\]
be the vector space of signed trace differences.  If $P,Q\in\Delta(\Om_T)$, then $P-Q\in V_T$.  For a finite audience basis $\F_T=\{f_1,\ldots,f_m\}$ define the measurement operator
\[
    \Phi_{\F_T}:\Delta(\Om_T)\to\R^m,
    \qquad
    \Phi_{\F_T}(P)=\bigl(\mathbb E_P f_1,\ldots,\mathbb E_P f_m\bigr).
\]
On signed differences write
\[
    \|\mu\|_{\F_T}=\|\Phi_{\F_T}(\mu)\|_\infty
    =\max_{1\leq i\leq m}|\mathbb E_\mu f_i|.
\]
This is a seminorm.  Its kernel is the set of signed residues that the audience basis cannot read.

\begin{theorem}[Measurement-polytope equivalence]
Let $C,D\subseteq\Delta(\Om_T)$ be compact convex trace sets and let $\F_T=\{f_1,\ldots,f_m\}$ be a finite audience basis.  Define the projected polytopes
\[
    K_C=\Phi_{\F_T}(C),
    \qquad
    K_D=\Phi_{\F_T}(D)
    \subseteq\R^m.
\]
Then
\[
    \inf_{P\in C,Q\in D}\Delta_{\F_T}(P,Q)
    =
    \operatorname{dist}_\infty(K_C,K_D),
\]
where
\[
    \operatorname{dist}_\infty(K_C,K_D)
    =\inf_{x\in K_C,y\in K_D}\|x-y\|_\infty.
\]
Consequently, perfect choric masking for this audience is equivalent to
\[
    K_C\cap K_D\neq\varnothing,
\]
and $\eps$-masking is equivalent to the two projected polytopes having $\ell_\infty$ distance at most $\eps$.
\end{theorem}

\begin{proof}
For $P\in C$ and $Q\in D$,
\[
    \Delta_{\F_T}(P,Q)
    =\max_i|\mathbb E_P f_i-\mathbb E_Q f_i|
    =\|\Phi_{\F_T}(P)-\Phi_{\F_T}(Q)\|_\infty.
\]
Taking the infimum over $P,Q$ gives exactly the distance between the images.  The zero-distance statement follows because $K_C,K_D$ are compact; zero distance is therefore intersection.  The approximate statement is the same identity with threshold $\eps$.
\end{proof}

\begin{theorem}[Finite separation certificate]
Suppose $K_C,K_D\subseteq\R^m$ are the projected polytopes above and
\[
    \operatorname{dist}_\infty(K_C,K_D)>\eps.
\]
Then there exists a vector $\lambda\in\R^m$ with $\|\lambda\|_1\leq1$ and a scalar $a\in\R$ such that, after orienting the two polytopes,
\[
    \lambda\cdot x\geq a+\eps
    \quad(x\in K_C),
    \qquad
    \lambda\cdot y\leq a
    \quad(y\in K_D).
\]
Equivalently, the signed audience score
\[
    g_\lambda=\sum_{i=1}^m\lambda_i f_i
\]
separates the two trace sets by margin at least $\eps$:
\[
    \inf_{P\in C}\mathbb E_P g_\lambda
    -
    \sup_{Q\in D}\mathbb E_Q g_\lambda
    \geq\eps.
\]
If the original $f_i$ take values in $[0,1]$, then $g_\lambda$ has oscillation at most one after subtracting a constant, so it is a bounded public certificate.
\end{theorem}

\begin{proof}
The closed $\ell_\infty$ $\eps$-neighborhood of $K_D$ is convex and disjoint from $K_C$.  By finite-dimensional hyperplane separation, there are $\lambda\neq0$ and $a$ with $\lambda\cdot x>a\geq\lambda\cdot y$ for $x\in K_C$ and $y$ in that neighborhood.  The support function of the $\ell_\infty$ unit ball is the $\ell_1$ norm, so requiring separation from the whole $\eps$-neighborhood normalizes the margin by $\eps\|\lambda\|_1$.  Scaling gives $\|\lambda\|_1\leq1$ and the displayed inequalities.  Substituting $x=\Phi_{\F_T}(P)$ and $y=\Phi_{\F_T}(Q)$ gives the signed audience score.
\end{proof}

\begin{corollary}[Kernel obstruction]
Let $C,D$ be compact convex trace sets.  Perfect masking for a finite audience basis holds if and only if
\[
    (C-D)\cap\Ker\Phi_{\F_T}\neq\varnothing,
\]
where $C-D=\{P-Q:P\in C,Q\in D\}$.  Approximate masking means that $C-D$ intersects the $\eps$-tube around this kernel in the seminorm $\|\cdot\|_{\F_T}$.
\end{corollary}

\begin{proof}
The equality $\Phi_{\F_T}(P)=\Phi_{\F_T}(Q)$ is equivalent to $P-Q\in\Ker\Phi_{\F_T}$.  The approximate statement is the definition of the seminorm distance to the kernel.
\end{proof}

\begin{definition}[Layer-localization certificate]
Let $L$ be a finite layer alphabet and suppose each layer value $\ell\in L$ has a trace polytope $C_\ell\subseteq\Delta(\Om_T)$.  The value $\ell_0$ is $(\F_T,\gamma)$-localizable against a comparison set $B\subseteq L\setminus\{\ell_0\}$ when
\[
    \operatorname{dist}_\infty\left(
       \Phi_{\F_T}(C_{\ell_0}),
       \conv\bigcup_{\ell\in B}\Phi_{\F_T}(C_\ell)
    \right)
    \geq\gamma.
\]
If this distance is zero, the audience can at most certify a group residue, not a layer-local residue.
\end{definition}

\begin{proposition}[Group residue is not localization]
Assume a finite audience basis.  If
\[
    \Phi_{\F_T}(C_{\ell_0})
    \cap
    \conv\bigcup_{\ell\in B}\Phi_{\F_T}(C_\ell)
    \neq\varnothing,
\]
then every separating statement that is true for the group convex hull is compatible with a trace mixture in which the carrier layer is not $\ell_0$.  Therefore the observation can certify at most that some member of the comparison group carries the residue; it cannot certify $\ell_0$ as its source.
\end{proposition}

\begin{proof}
Choose a common projected point $z$.  There is $P\in C_{\ell_0}$ with $\Phi(P)=z$ and a convex combination $Q=\sum_j a_jQ_j$, with each $Q_j\in C_{\ell_j}$ for $\ell_j\in B$, such that $\Phi(Q)=z$.  Every audience test in the basis has the same expectation under $P$ and under the mixture $Q$.  Thus any conclusion depending only on those expectations remains true when the same projected observation is generated by the comparison mixture.  The source layer is not identified.
\end{proof}

\begin{theorem}[Information lower bound for localization]
Let a carrier index $\Sigma$ be uniform on a cover cell of size $m\geq2$, and let $Z$ be the entire observation available to the response policy in one window.  For any estimator $\widehat\Sigma(Z)$,
\[
    \Pr\{\widehat\Sigma\neq\Sigma\}
    \geq
    1-\frac{\MI(\Sigma;Z)+\log 2}{\log m}.
\]
Consequently, to guarantee correct localization probability at least $q$, the observation channel must carry
\[
    \MI(\Sigma;Z)\geq q\log m-\log 2.
\]
A room that increases cover size without increasing carrier information cannot turn pressure into reliable localization.
\end{theorem}

\begin{proof}
The first inequality is Fano's inequality in finite form~\cite{cover2006}.  Indeed,
\[
    \Ent(\Sigma\mid Z)=\Ent(\Sigma)-\MI(\Sigma;Z)=\log m-\MI(\Sigma;Z).
\]
If $P_e=\Pr\{\widehat\Sigma\neq\Sigma\}$, then
\[
    \Ent(\Sigma\mid Z)\leq \log 2+P_e\log m,
\]
using the standard bound $h(P_e)\leq\log2$ and $\log(m-1)\leq\log m$.  Rearranging gives the displayed lower bound on $P_e$.  If the success probability is at least $q$, then $P_e\leq1-q$, so the same inequality gives $\MI(\Sigma;Z)\geq q\log m-\log2$.
\end{proof}

\begin{theorem}[Attention is kernel enlargement]
Let $L^{\rm att}:\Om_T\to\Delta(\widetilde\Om_T)$ be an attention lens and let $\widetilde\Phi$ be the measurement operator of the post-attention audience.  The effective audience operator is
\[
    \Phi^{\rm eff}=\widetilde\Phi L^{\rm att}.
\]
For any two trace laws $P,Q$,
\[
    \|P-Q\|_{\Phi^{\rm eff}}
    \leq
    \TV(P,Q).
\]
Moreover, if
\[
    P-Q\in\Ker(\widetilde\Phi L^{\rm att})
    \quad\text{but}\quad
    P-Q\notin\Ker(\Phi_{\F_T}),
\]
then the pair is exposed for the full audience basis but perfectly masked for the attention-filtered audience.  Thus attention loss is not a metaphor; it is membership in a larger nullspace.
\end{theorem}

\begin{proof}
The equality $\Phi^{\rm eff}=\widetilde\Phi L^{\rm att}$ is the definition of first applying the Markov kernel and then measuring.  Markov kernels contract total variation, and bounded tests are dominated by total variation, so the displayed contraction follows.  The kernel statement is immediate: the post-attention measurements coincide exactly when the signed difference lies in the kernel of the composed operator.  If the same difference is not in the full-audience kernel, some full trace test separates it.
\end{proof}

\begin{proposition}[Bernoulli aperture identity]
Let $C\subseteq\Om_T$ be a cylinder and suppose the attention lens reports the event only when an independent aperture variable $D_C\sim\operatorname{Bernoulli}(p_C)$ hits it.  If the reported indicator is $\widetilde{\one_C}=D_C\one_C$, then for any laws $P,Q$,
\[
    \mathbb E_P\widetilde{\one_C}-\mathbb E_Q\widetilde{\one_C}
    =p_C\bigl(P(C)-Q(C)\bigr).
\]
A genuine cylinder residue therefore becomes invisible at the rate at which its aperture probability goes to zero.
\end{proposition}

\begin{proof}
Independence gives $\mathbb E_P[D_C\one_C]=p_CP(C)$ and similarly for $Q$.  Subtract.
\end{proof}

\begin{theorem}[Debt-hazard decoupling]
Let a repeated room have one-window hazard
\[
    h(D)=1-\exp(-\alpha(D)r(D))
\]
for a fixed carrier, where $\alpha(D)>0$ and $r(D)>0$ are differentiable functions of debt on an interval.  Then
\[
    \frac{d}{dD}h(D)>0
    \quad\Longleftrightarrow\quad
    \frac{d}{dD}\log\bigl(\alpha(D)r(D)\bigr)>0.
\]
Debt alone is therefore not a hazard certificate.  Debt becomes carrier hazard only through an increase in carrier-local attention, carrier-local residue, or both.
\end{theorem}

\begin{proof}
Differentiate:
\[
    h'(D)=\exp(-\alpha r)\,\frac{d}{dD}(\alpha r).
\]
The exponential factor is positive, so the sign is the sign of $(\alpha r)'$.  Since $\alpha r>0$, this is the sign of $(\log(\alpha r))'$.  The interpretation follows because the product contains only carrier-local attention and carrier-local residue.
\end{proof}

\section{Choric masking}

\begin{definition}[Choric mask]
Let $\F_T$ be an audience test class. Hidden loci $\theta,\theta'\in\Theta$ are $(\F_T,\eps)$-chorically masked at horizon $T$ if there exist laws
\[
    P\in \operatorname{Ch}_T(\theta),\qquad Q\in\operatorname{Ch}_T(\theta')
\]
such that
\[
    \Delta_{\F_T}(P,Q)\leq \eps.
\]
They are perfectly chorically masked if this holds with $\eps=0$.
\end{definition}

This definition does not require one participant to imitate another. It requires the two hidden loci to have compatible public mixtures. The distinction is important. A public environment can be safe without every individual acting alike, provided the aggregate trace available to the audience has no separating residue.

\begin{definition}[Exposure radius]
The exposure radius of $\theta$ against a comparison set $B\subseteq\Theta$ is
\[
    \mathfrak r_T(\theta;B,\F_T)
    =\inf_{P\in\operatorname{Ch}_T(\theta)}\ \sup_{\theta'\in B}\ \inf_{Q\in\operatorname{Ch}_T(\theta')}\Delta_{\F_T}(P,Q).
\]
When $B=\Theta\setminus\{\theta\}$, write $\mathfrak r_T(\theta;\F_T)$.
\end{definition}

The exposure radius is the smallest ambient distance that remains after the locus chooses its best chorus. If it is zero, the locus can enter the room without leaving a testable role residue against all other loci. If it is positive, there is an irreducible public asymmetry.

\begin{definition}[Room radius]
The room radius is
\[
    \mathfrak R_T(\A,\F_T)=\sup_{\theta\in\Theta}\mathfrak r_T(\theta;\F_T).
\]
The room is $\eps$-level at horizon $T$ when $\mathfrak R_T(\A,\F_T)\leq\eps$.
\end{definition}

The room radius is a social parameter as much as a mathematical one. It measures whether the public grammar of the room makes some hidden loci inherently louder than others.

\section{Stratified rooms and localization}

The previous definitions are enough for a two-name comparison, but many public rooms are not two-name rooms. A visible trace may be produced by several hidden layers at once: a protected person, an object category, a companion relation, a queue cohort, a sensor condition, a clock phase, and an auxiliary public label. A residue can be present in the room without being localizable to the protected layer. This section makes that distinction explicit.

\begin{definition}[Stratified hidden carrier]
A stratified carrier is a finite product
\[
    \Lambda=\Lambda_0\times\Lambda_1\times\cdots\times\Lambda_d.
\]
A point \(\lambda=(\lambda_0,\ldots,\lambda_d)\in\Lambda\) is a layered hidden state. The coordinate \(\lambda_j\) is called layer \(j\). The earlier model is recovered by taking \(d=0\) and \(\Theta=\Lambda_0\).
\end{definition}

In a stratified room the transition and observation kernels may depend on the whole layered state. For example, one coordinate may describe a bearer, another an object class, another a companion relation, another a clock phase, and another a sensor condition. The audience still sees only \(Y_{1:T}\). The layer names are analytic coordinates, not labels shown to the audience.

\begin{definition}[Layer choric hull]
Let an ambient release system have hidden carrier \(\Lambda\). Fix a layer \(j\) and a layer value \(a\in\Lambda_j\). The layer choric hull of \(a\) is
\[
    C_T^{(j)}(a)
    =\conv\{P_{\lambda,r}^{\pi,T}:\lambda_j=a,\ \lambda\in\Lambda,\ r\in R,\ \pi\in\Pi_T(r)\}.
\]
For a set \(U\subseteq\Lambda_j\), put
\[
    C_T^{(j)}(U)=\conv\bigcup_{a\in U}C_T^{(j)}(a).
\]
\end{definition}

This is the first place where the model becomes more than a protected-versus-unprotected comparison. A layer value is not compared only with one opposite value. It is compared with every public trace that can be generated by all nuisance layers while that coordinate is fixed.

\begin{definition}[Measurement map]
Let \(\G_T=\{G_1,\ldots,G_k\}\) be a finite family of public events. The measurement map is
\[
    \Phi_{\G_T}:\Delta(Y^T)\to[0,1]^k,
    \qquad
    \Phi_{\G_T}(P)=(P(G_1),\ldots,P(G_k)).
\]
The measurement hull of layer value \(a\) is
\[
    H_T^{(j)}(a;\G_T)=\Phi_{\G_T}\big(C_T^{(j)}(a)\big)\subseteq[0,1]^k.
\]
\end{definition}

If the audience can only use the events in \(\G_T\), then it does not see the whole trace polytope. It sees only a projected polytope in the measurement cube. Two layer values may be far apart under full traces and indistinguishable under \(\Phi_{\G_T}\). Conversely, a single well-chosen event can split two values even if their raw traces look informally similar.

\begin{proposition}[Measurement reduction]
Let \(\F(\G_T)\) be the test class generated by the indicators \(1_{G_1},\ldots,1_{G_k}\) and their complements. Then for laws \(P,Q\),
\[
    \Delta_{\F(\G_T)}(P,Q)
    =\|\Phi_{\G_T}(P)-\Phi_{\G_T}(Q)\|_\infty.
\]
Consequently, for layer values \(a,b\in\Lambda_j\),
\[
    \inf_{P\in C_T^{(j)}(a),\ Q\in C_T^{(j)}(b)}
    \Delta_{\F(\G_T)}(P,Q)
    =\operatorname{dist}_\infty\big(H_T^{(j)}(a;\G_T),H_T^{(j)}(b;\G_T)\big).
\]
\end{proposition}

\begin{proof}
For an event indicator, the expectation gap is exactly \(P(G_i)-Q(G_i)\), and complements only change the sign. Taking the maximum over the listed events gives the \(\ell_\infty\) norm of the measurement difference. The second display follows by minimizing this quantity over the two layer hulls and using the definition of the image sets under \(\Phi_{\G_T}\).
\end{proof}

\begin{definition}[Localization set]
For a public law \(P\), layer \(j\), audience class \(\F_T\), and tolerance \(\eps\geq0\), define
\[
    \operatorname{Loc}_{j,\eps}^{\F_T}(P)
    =\left\{a\in\Lambda_j:
      \inf_{Q\in C_T^{(j)}(a)}\Delta_{\F_T}(P,Q)\leq\eps
    \right\}.
\]
This is the set of layer values that remain compatible with \(P\) at resolution \(\eps\).
\end{definition}

A residue is not the same as a localization. The audience may know that the room differs from a baseline and still not know which layer caused the difference. The following theorem is the formal version of that statement.

\begin{theorem}[Residue without localization]
Fix layer \(j\), a finite event family \(\G_T\), and a public law \(P\). If
\[
    \Phi_{\G_T}(P)\in H_T^{(j)}(a;\G_T)
\]
for \(m\) distinct values \(a\in\Lambda_j\), then every audience whose tests factor through \(\Phi_{\G_T}\) has a zero-tolerance localization set of size at least \(m\):
\[
    \left|\operatorname{Loc}_{j,0}^{\F(\G_T)}(P)\right|\geq m.
\]
More generally, if \(\Phi_{\G_T}(P)\) is within \(\eps\) in \(\ell_\infty\) distance of the measurement hulls of \(m\) values, then
\[
    \left|\operatorname{Loc}_{j,\eps}^{\F(\G_T)}(P)\right|\geq m.
\]
\end{theorem}

\begin{proof}
If \(\Phi_{\G_T}(P)\in H_T^{(j)}(a;\G_T)\), then there exists \(Q_a\in C_T^{(j)}(a)\) such that \(\Phi_{\G_T}(Q_a)=\Phi_{\G_T}(P)\). By measurement reduction, \(\Delta_{\F(\G_T)}(P,Q_a)=0\). Hence \(a\in\operatorname{Loc}_{j,0}^{\F(\G_T)}(P)\). Repeating this for the \(m\) values gives the first claim. The approximate claim is identical with \(\ell_\infty\) distance at most \(\eps\).
\end{proof}

\begin{definition}[Cover multiplicity]
For a layer \(j\), event family \(\G_T\), and measurement vector \(v\in[0,1]^k\), the cover multiplicity is
\[
    m_T^{(j)}(v;\G_T)
    =\left|\{a\in\Lambda_j:v\in H_T^{(j)}(a;\G_T)\}\right|.
\]
The \(\eps\)-cover multiplicity \(m_{T,\eps}^{(j)}(v;\G_T)\) counts the values whose measurement hulls lie within \(\eps\) of \(v\) in \(\ell_\infty\) distance.
\end{definition}

Cover multiplicity is the number of voices that can carry the same public measurement. It generalizes the size of an anonymity set~\cite{chaum1981,reiter1998,serjantov2002} from users behind a relay to layer values behind a measurement hull. A high value means the residue is spread across a chorus. A low value means the residue is almost a name.

\begin{corollary}[Carrier lower bound]
Let \(P\) be the public trace law of a room and suppose the audience tests factor through \(\Phi_{\G_T}\). Then no sound localization certificate for layer \(j\) at tolerance \(\eps\) can return fewer than
\[
    m_{T,\eps}^{(j)}\big(\Phi_{\G_T}(P);\G_T\big)
\]
layer values without discarding at least one value that is compatible with the observed public measurements.
\end{corollary}

\begin{proof}
By the residue-without-localization theorem, each value counted by \(m_{T,\eps}^{(j)}\) has a choric law whose measurements are \(\eps\)-close to those of \(P\). A certificate returning fewer values must omit at least one such compatible value.\end{proof}

\begin{definition}[Choric core]
For a set \(U\subseteq\Lambda_j\), the measurement core is
\[
    \operatorname{Core}_T^{(j)}(U;\G_T)
    =\bigcap_{a\in U}H_T^{(j)}(a;\G_T).
\]
If the core is nonempty, the values in \(U\) admit a common measurement under the audience \(\G_T\).
\end{definition}

\begin{theorem}[Finite Helly certificate]
Let \(|\G_T|=k\). For a finite set \(U\subseteq\Lambda_j\), if every subfamily \(V\subseteq U\) with \(|V|\leq k+1\) has nonempty measurement core, then
\[
    \operatorname{Core}_T^{(j)}(U;\G_T)\neq\varnothing.
\]
Thus a common group mask for an event audience with \(k\) measurements can be certified by checking only \((k+1)\)-wise intersections of measurement hulls.
\end{theorem}

\begin{proof}
Each measurement hull \(H_T^{(j)}(a;\G_T)\) is a convex subset of \(\mathbb R^k\), because it is the linear image of a convex polytope. The claim is exactly Helly's theorem~\cite{helly1923} in \(\mathbb R^k\) applied to the finite family \(\{H_T^{(j)}(a;\G_T):a\in U\}\).\end{proof}

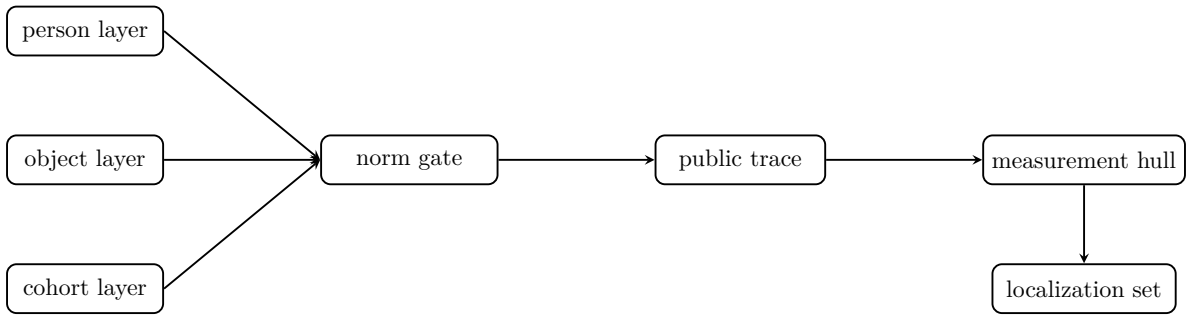
\begin{figure}[H]
\centering
\resizebox{0.98\linewidth}{!}{%
\begin{tikzpicture}[node distance=1.15cm, >=stealth, thick, every node/.style={font=\small}]
\node[draw, rounded corners, minimum width=2.3cm, minimum height=.72cm] (p) {person layer};
\node[draw, rounded corners, below=of p, minimum width=2.3cm, minimum height=.72cm] (o) {object layer};
\node[draw, rounded corners, below=of o, minimum width=2.3cm, minimum height=.72cm] (c) {cohort layer};
\node[draw, rounded corners, right=2.3cm of o, minimum width=2.6cm, minimum height=.72cm] (gate) {norm gate};
\node[draw, rounded corners, right=2.3cm of gate, minimum width=2.5cm, minimum height=.72cm] (trace) {public trace};
\node[draw, rounded corners, right=2.3cm of trace, minimum width=2.7cm, minimum height=.72cm] (meas) {measurement hull};
\node[draw, rounded corners, below=of meas, minimum width=2.7cm, minimum height=.72cm] (loc) {localization set};
\draw[->] (p.east) -- (gate.west);
\draw[->] (o.east) -- (gate.west);
\draw[->] (c.east) -- (gate.west);
\draw[->] (gate) -- (trace);
\draw[->] (trace) -- (meas);
\draw[->] (meas) -- (loc);
\end{tikzpicture}%
}
\caption{A stratified room. A real residue may appear in the measurement hull while the localization set remains wide across person, object, or cohort layers.}
\label{fig:stratified-room}
\end{figure}

\section{Bounded attention and perceptual apertures}

The previous sections treat the audience as a class of tests on the full public trace. This is already relative to an audience, but it still grants the audience a complete trace before the test is applied. Many rooms are not experienced that way. A human observer, a crowded desk, or a small review group receives a thinned trace: some coordinates are missed, some are delayed, some are merged with nearby events, and some are encoded only as a coarse impression. The distinction is not cosmetic. A residue can exist in the public trace while failing to exist inside the observer's effective aperture.

\begin{definition}[Attention lens]
Fix a horizon \(T\). An attention lens is a Markov kernel
\[
    L:Y^T\to \Delta(\widetilde Y^T),
\]
where \(\widetilde Y\) is a finite perceived alphabet. For a law \(P\in\Delta(Y^T)\), write \(LP\) for the pushforward law on \(\widetilde Y^T\):
\[
    (LP)(z)=\sum_{y\in Y^T}P(y)L(z\mid y).
\]
An attentive audience is a pair \((L,\widetilde\F_T)\), where \(\widetilde\F_T\subseteq[0,1]^{\widetilde Y^T}\) is a test class on perceived traces.
\end{definition}

The lens is not an error term added after the model. It is part of the public reading process. It may represent gaze direction, dwell time, memory compression, channel delay, queue density, social distraction, or the fact that several candidates are inspected at once. The observer never applies \(f\) to \(Y_{1:T}\). The observer applies \(f\) to \(\widetilde Y_{1:T}\).

\begin{definition}[Attention-filtered distance]
For an attentive audience \((L,\widetilde\F_T)\), define
\[
    \Delta_{L,\widetilde\F_T}(P,Q)
    :=\Delta_{\widetilde\F_T}(LP,LQ)
    =\sup_{f\in\widetilde\F_T}\left|\mathbb E_{LP}f-\mathbb E_{LQ}f\right|.
\]
The associated lifted full-trace test class is
\[
    L^\ast\widetilde\F_T=\{f\circ L:f\in\widetilde\F_T\},
\]
where \((f\circ L)(y)=\sum_z f(z)L(z\mid y)\). Thus
\[
    \Delta_{L,\widetilde\F_T}(P,Q)=\Delta_{L^\ast\widetilde\F_T}(P,Q).
\]
\end{definition}

\begin{proposition}[Data processing for audience distance]
For every attention lens \(L\), test class \(\widetilde\F_T\), and laws \(P,Q\),
\[
    \Delta_{L,\widetilde\F_T}(P,Q)\leq \TV(P,Q).
\]
If \(L^\ast\widetilde\F_T\subseteq \F_T\), then
\[
    \Delta_{L,\widetilde\F_T}(P,Q)\leq \Delta_{\F_T}(P,Q).
\]
\end{proposition}

\begin{proof}
The first inequality follows because every lifted test \(f\circ L\) is a bounded function from \(Y^T\) to \([0,1]\), and total variation is the supremum over all such tests. The second follows by inclusion of the lifted class in \(\F_T\).
\end{proof}

The proposition says that attention cannot create more statistical separation than a full observer with access to all bounded tests. It can, however, destroy a separation that exists at the full-trace level. This is the missing layer between ``the trace contains a residue'' and ``the observer can use the residue.'' The inequality is the audience-relative form of the data-processing law familiar from information-flow analysis~\cite{smith2009,alvim2020}.

\begin{definition}[Attentional choric masking]
Two laws \(P,Q\in\Delta(Y^T)\) are \((L,\widetilde\F_T,\eps)\)-attention-masked when
\[
    \Delta_{L,\widetilde\F_T}(P,Q)\leq \eps.
\]
They are \((\F_T,\eta;L,\widetilde\F_T,\eps)\)-split when
\[
    \Delta_{\F_T}(P,Q)>\eta
    \qquad\text{and}\qquad
    \Delta_{L,\widetilde\F_T}(P,Q)
    \leq\eps .
\]
In a split pair, exposure exists for the full trace audience but not for the bounded attentive audience.
\end{definition}

\begin{definition}[Budgeted gaze lens]
Let \(\mathcal I\) be a finite family of coordinate subsets of \(\{1,\ldots,T\}\). A budgeted gaze profile is a probability vector \(\beta\in\Delta(\mathcal I)\). The corresponding lens first samples \(I\sim\beta\), reveals \((y_i)_{i\in I}\), and replaces all other coordinates by a blank symbol \(\bot\). The perceived alphabet is \(\widetilde Y=Y\cup\{\bot\}\).
\end{definition}

This lens is intentionally simple. It turns attention into a stochastic aperture. A larger dwell budget places more mass on larger coordinate sets. A fragmented room places mass on many small coordinate sets. The observer's test may be sharp on what is seen and completely powerless on what is blanked.

\begin{definition}[Diagnostic support]
An event \(G\subseteq Y^T\) is supported by coordinates \(J\subseteq\{1,\ldots,T\}\) if membership in \(G\) is determined by \((y_j)_{j\in J}\). For a gaze profile \(\beta\), its hit probability on \(J\) is
\[
    h_\beta(J)=\beta\{I\in\mathcal I:J\subseteq I\}.
\]
\end{definition}

\begin{theorem}[Aperture bound for a cylinder residue]
Let \(G\subseteq Y^T\) be supported by \(J\). Let \(L_\beta\) be a budgeted gaze lens. Let \(\widetilde G\subseteq\widetilde Y^T\) be any perceived event that can agree with \(G\) only when all coordinates in \(J\) are visible. Then for all laws \(P,Q\),
\[
    \left|(L_\beta P)(\widetilde G)-(L_\beta Q)(\widetilde G)\right|
    \leq h_\beta(J)\, |P(G)-Q(G)| + \bigl(1-h_\beta(J)\bigr)\,\kappa(\widetilde G,J),
\]
where \(\kappa(\widetilde G,J)\) is the largest possible perceived imbalance contributed by non-hit apertures. In particular, if non-hit apertures make \(\widetilde G\) independent of the distinction between \(P\) and \(Q\), then
\[
    \left|(L_\beta P)(\widetilde G)-(L_\beta Q)(\widetilde G)\right|
    \leq h_\beta(J)\, |\res_{P,Q}(G)|.
\]
\end{theorem}

\begin{proof}
Decompose the lens by the sampled aperture:
\[
    L_\beta=\sum_{I\in\mathcal I}\beta(I)L_I.
\]
Split the sum into apertures with \(J\subseteq I\) and apertures with \(J\not\subseteq I\). On hit apertures, the perceived event can carry at most the full residue of the diagnostic event, hence the contribution is bounded by \(h_\beta(J)|P(G)-Q(G)|\). On non-hit apertures, at least one determining coordinate is blanked; by definition their total possible imbalance is bounded by \((1-h_\beta(J))\kappa(\widetilde G,J)\). If the non-hit perceived law no longer depends on the distinction, this second term is zero.
\end{proof}

\begin{corollary}[Exposure without perception]
There exist laws \(P,Q\), a full-trace event audience \(\F_T\), and a budgeted gaze lens \(L_\beta\) such that
\[
    \Delta_{\F_T}(P,Q)>0
    \qquad\text{but}\qquad
    \Delta_{L_\beta,\widetilde\F_T}(P,Q)=0.
\]
\end{corollary}

\begin{proof}
Let \(P\) and \(Q\) differ only on a coordinate set \(J\). Let the full audience contain the indicator of a cylinder supported by \(J\). Choose a gaze profile with \(h_\beta(J)=0\) and choose perceived tests that depend only on the revealed coordinates. Then the full audience separates the laws, while the perceived trace distributions are identical.
\end{proof}

This corollary is the clean mathematical form of an attention shadow. The residue is not absent. It is outside the aperture.

\begin{definition}[Layer attention shadow]
For a stratified carrier \(\Lambda\), layer \(j\), value \(a\in\Lambda_j\), and public law \(P\), define the attention shadow of \(a\) at tolerance \(\eps\) by
\[
    \operatorname{Sh}_{j,\eps}^{L,\widetilde\F_T}(P)
    =\left\{a\in\Lambda_j:
      \inf_{Q\in C_T^{(j)}(a)}\Delta_{L,\widetilde\F_T}(P,Q)\leq\eps
    \right\}.
\]
A value lies in the shadow when it remains compatible after attention filtering, even if it is not compatible under a stronger full-trace audience.
\end{definition}

\begin{theorem}[Attention dilution across simultaneous candidates]
Consider \(n\) candidate bearers and a bounded observer with total gaze budget \(B>0\). Suppose the observer assigns fractions \(\alpha_1,\ldots,\alpha_n\) with \(\sum_i\alpha_i=1\). For bearer \(i\), suppose every diagnostic event that would separate its anomalous-bearing hull from the ordinary union requires a coordinate set \(J_i\), and the gaze profile satisfies
\[
    h_{\beta_i}(J_i)\leq c\alpha_i B
\]
for a constant \(c\) determined by the sampling geometry. If the non-hit apertures carry no diagnostic imbalance, then every perceived event test for bearer \(i\) has separation at most
\[
    c\alpha_i B\,\delta_i,
\]
where \(\delta_i\) is the corresponding full-trace diagnostic residue. Consequently, when attention is split across many compatible candidates, the perceived residue of any one candidate can fall below threshold even when the full-trace residue is positive.
\end{theorem}

\begin{proof}
Apply the aperture bound to bearer \(i\). The hit probability is at most \(c\alpha_iB\), and the non-hit term is zero by assumption. The perceived separation is therefore at most \(c\alpha_iB\delta_i\). If a decision threshold \(\tau\) satisfies \(c\alpha_iB\delta_i\leq\tau\), the bounded observer cannot certify the diagnostic event for that bearer at threshold \(\tau\), even though \(\delta_i>0\) on the full trace.
\end{proof}

\begin{definition}[Attentional room radius]
The attention-filtered exposure radius of locus \(\theta\) against a comparison set \(B\subseteq\Theta\setminus\{\theta\}\) is
\[
    \mathfrak r_T^{L}(\theta;B,\widetilde\F_T)
    =\inf_{P\in\operatorname{Ch}_T(\theta)}
      \inf_{Q\in\operatorname{Ch}_T(B)}
      \Delta_{L,\widetilde\F_T}(P,Q).
\]
The attention-filtered room radius is
\[
    \mathfrak R_T^{L}(\A,\widetilde\F_T)
    =\sup_{\theta\in\Theta}\mathfrak r_T^{L}(\theta;\Theta\setminus\{\theta\},\widetilde\F_T).
\]
\end{definition}

The inequality \(\mathfrak R_T^L\leq \mathfrak R_T\) is not a promise of safety. It is a warning about measurement. A system may look safe because the observer lacks enough aperture, not because the full public trace lacks exposure.

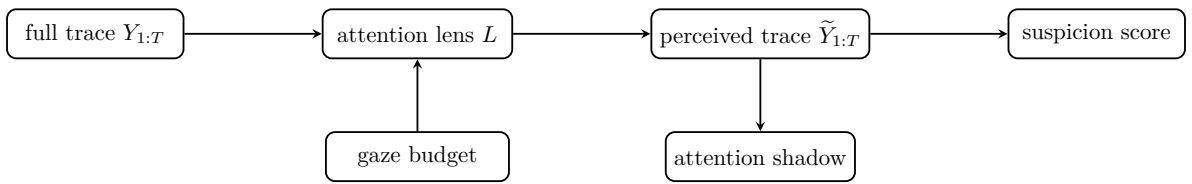
\begin{figure}[H]
\centering
\resizebox{0.98\linewidth}{!}{%
\begin{tikzpicture}[node distance=1.25cm, >=stealth, thick, every node/.style={font=\small}]
\node[draw, rounded corners, minimum width=2.7cm, minimum height=.75cm] (trace) {full trace $Y_{1:T}$};
\node[draw, rounded corners, right=2.1cm of trace, minimum width=2.9cm, minimum height=.75cm] (lens) {attention lens $L$};
\node[draw, rounded corners, right=2.1cm of lens, minimum width=2.9cm, minimum height=.75cm] (ptrace) {perceived trace $\widetilde Y_{1:T}$};
\node[draw, rounded corners, right=2.1cm of ptrace, minimum width=2.7cm, minimum height=.75cm] (score) {suspicion score};
\node[draw, rounded corners, below=1.1cm of lens, minimum width=2.8cm, minimum height=.75cm] (budget) {gaze budget};
\node[draw, rounded corners, below=1.1cm of ptrace, minimum width=2.8cm, minimum height=.75cm] (shadow) {attention shadow};
\draw[->] (trace) -- (lens);
\draw[->] (lens) -- (ptrace);
\draw[->] (ptrace) -- (score);
\draw[->] (budget) -- (lens);
\draw[->] (ptrace) -- (shadow);
\end{tikzpicture}%
}
\caption{Attention-filtered reading. A residue may be present in the full trace and absent from the perceived trace after gaze thinning, dwell limits, and budget splitting.}
\label{fig:attention-lens}
\end{figure}

\section{Inspection hall chamber}

This section gives a finite chamber in which the residue is real but not uniquely borne. The chamber is stylized and analytic: it measures when a public anomaly localizes and when it remains distributed across ordinary trace sources. The chamber now has two readings: the full public trace and the attention-filtered trace available to a bounded observer.

Let the hidden carrier be
\[
    \Lambda=B\times O\times C\times U,
\]
where \(B\) is a bearer layer, \(O\) is an object-class layer, \(C\) is a companion-or-cohort layer, and \(U\) is an unmodeled-condition layer. The audience sees a trace
\[
    Y_{1:T}=(\text{queue phase},\text{bag image class},\text{odor band},\text{pause band},\text{cohort motion},\text{officer contact})_{1:T}
\]
after the room coarsens all raw measurements into a finite alphabet. The audience event family is
\[
    \G_T=\{G_1,G_2,G_3,G_4\},
\]
where \(G_1\) is an odor-band event, \(G_2\) an organic-image event, \(G_3\) a pause-band event, and \(G_4\) a cohort-motion mismatch event. These names are only public event names. The model never assumes that an event is a definitive cause.

The layer convention used in the chamber is as follows. The bearer coordinate $b\in B$ is the candidate carrier of a trace source. The object coordinate $o\in O$ records the public object class, such as ordinary, food-like, medicine-like, gift-like, empty, or anomalous. The cohort coordinate $c\in C$ represents nearby ordinary sources and companion motion. The condition coordinate $u\in U$ records clock phase, queue density, sensor condition, and other environmental variables. The audience does not see these coordinates directly; it only applies the finite measurement family $\G_T$ to the public trace. Localization therefore means recovering a layer value through its measurement hull, not naming a person by fiat.

For bearer \(b\), define the anomalous-bearing hull
\[
    A_T(b)=\conv\{P_{(b,o,c,u),r}^{\pi,T}:o\in O_{\rm an},\ c\in C,\ u\in U,\ r\in R,\ \pi\in\Pi_T(r)\},
\]
and the ordinary-bearing hull
\[
    N_T(b)=\conv\{P_{(b,o,c,u),r}^{\pi,T}:o\in O_{\rm ord},\ c\in C,\ u\in U,\ r\in R,\ \pi\in\Pi_T(r)\}.
\]
Here \(O_{\rm an}\) and \(O_{\rm ord}\) are disjoint analytic classes. They need not be publicly named in the trace.

\begin{definition}[Inspection ambiguity set]
Given a public law \(P\), define the bearer ambiguity set at tolerance \(\eps\) by
\[
    \mathcal B_\eps(P)
    =\left\{b\in B:
      \inf_{Q\in A_T(b)\cup N_T(b)}
      \Delta_{\F(\G_T)}(P,Q)\leq\eps
    \right\}.
\]
The anomalous-bearing ambiguity set is
\[
    \mathcal B_\eps^{\rm an}(P)
    =\left\{b\in B:
      \inf_{Q\in A_T(b)}
      \Delta_{\F(\G_T)}(P,Q)\leq\eps
    \right\}.
\]
\end{definition}

The first set says which bearers can explain the public measurements by any ordinary or anomalous source. The second says which bearers can explain them as anomalous bearers. The gap between the two sets is the gap between detecting a residue and assigning it.

\begin{theorem}[Group-stop certificate]
Let \(P\) be the public law in the inspection chamber and put \(v=\Phi_{\G_T}(P)\). If \(v\) lies within \(\eps\) of the measurement hulls of \(q\) distinct bearers, then every sound \(\G_T\)-based selection rule that contains the true compatible bearer must select at least \(q\) bearers. Equivalently,
\[
    |\mathcal B_\eps(P)|\geq q.
\]
If \(v\) is within \(\eps\) of an anomalous-bearing hull for one bearer and within \(\eps\) of ordinary-bearing hulls for \(q-1\) other bearers, then the room has detected an anomalous residue without localizing an anomalous bearer by \(\G_T\) alone.
\end{theorem}

\begin{proof}
The first claim is the carrier lower bound applied to the bearer layer. For the second claim, the public measurement vector is compatible with one anomalous explanation and with \(q-1\) ordinary explanations. Since all tests factor through \(\Phi_{\G_T}\), no test in the audience class separates these explanations by more than \(\eps\). The residue is therefore a property of the measurement region, not of a unique bearer.\end{proof}

\begin{corollary}[Attentional group-stop certificate]
Let \(P\) be the public law in the inspection chamber and let \((L,\widetilde\F_T)\) be an attentive audience. If the perceived law \(LP\) lies within \(\eps\) of the attention-filtered hulls
\[
    L A_T(b)\cup L N_T(b)
\]
for \(q\) distinct bearers, then every sound attentive selection rule that contains all compatible bearers must select at least \(q\) bearers. If, in addition, the full trace audience has a diagnostic residue for one bearer while the attentive audience has radius at most \(\eps\), then the chamber has trace-level exposure but no attention-level bearer certificate.
\end{corollary}

\begin{proof}
Apply the carrier lower bound after replacing every trace law \(Q\) by its pushforward \(LQ\). The second statement is exactly the split condition: full-trace separation is positive, while attention-filtered distance to the compatible hull is within tolerance.
\end{proof}

\begin{proposition}[Attention shadow in the chamber]
Suppose a diagnostic event \(G\) for an anomalous-bearing hull is supported by coordinates \(J\), and a bounded observer's gaze profile has hit probability \(h_\beta(J)\leq h\). If non-hit apertures carry no diagnostic imbalance, then the attentive residue of \(G\) is at most
\[
    h\,\res^\star(G),
\]
where \(\res^\star(G)\) is the largest full-trace residue of \(G\) over the relevant anomalous-versus-ordinary comparison. Thus a chamber can have a real diagnostic event and still place its bearer in the attention shadow whenever \(h\,\res^\star(G)\) is below the room's selection threshold.
\end{proposition}

\begin{proof}
This is the aperture bound specialized to the anomalous-bearing and ordinary-bearing hulls of the chamber, followed by maximization over the relevant laws.
\end{proof}

\begin{proposition}[Unique sensor break]
Suppose there is an event \(G_\star\) such that for some bearer \(b_\star\),
\[
    P(G_\star)\geq a\quad\text{for all }P\in A_T(b_\star),
\]
while for every ordinary-bearing law \(Q\in N_T(b)\) and every \(b\in B\),
\[
    Q(G_\star)\leq b_0<a.
\]
If \(1_{G_\star}\) belongs to the audience class, then the best possible masking gap between \(A_T(b_\star)\) and the ordinary union is at least \(a-b_0\). Thus the choric mask breaks exactly when the audience obtains a public event whose attainable interval is separated from all ordinary cover intervals.
\end{proposition}

\begin{proof}
Apply the unique-gesture lower bound to \(G_\star\), with the anomalous-bearing hull on one side and the convex union of ordinary-bearing hulls on the other.\end{proof}

The chamber can certify several distinct facts, and the distinctions matter. If a measurement vector $v$ lies in both $H_T^{(j)}(a)$ and $H_T^{(j)}(b)$, then the observed geometry certifies no unique value at layer $j$. If $v$ lies near several hulls, the room has at most a group-level attention certificate. If an event $G$ has disjoint attainable intervals, $I(a,G)\cap I(b,G)=\varnothing$, then the event distinguishes the layer values. If a unique sensor event $G_\star$ has margin $a-b_0>0$, the mask breaks by that margin. If $\Delta_{\F}(P,Q)>\eta$ while $\Delta_{L,\widetilde\F}(P,Q)\leq\eps$, the trace contains exposure but the perceived trace does not. If a weak event satisfies $p>q$ across repeated windows, the residue has an amplification channel. Thus a residue may justify group selection without justifying a unique carrier claim.

The chamber explains a common public failure mode. The audience may be correct that ``something in this local flow is off'' while still being mathematically unable to attach the residue to a single bearer. The trace source sits in a choric hull shared by ordinary object classes, cohort motion, clock conditions, and sensor noise. A group selection is then not evidence of precise localization; it is the visible shape of a wide localization set.

\section{A map of the stage}

The following diagram is only a mnemonic. The formal object is the tuple $\A$ and the induced trace laws.

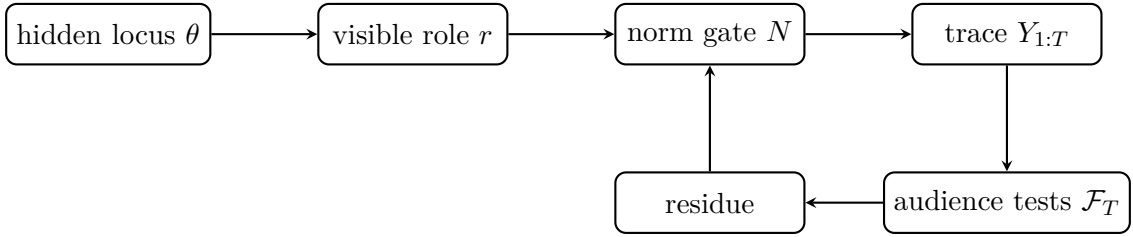
\begin{figure}[H]
\centering
\begin{tikzpicture}[node distance=1.4cm, >=stealth, thick]
\node[draw, rounded corners, minimum width=2.7cm, minimum height=.8cm] (theta) {hidden locus $\theta$};
\node[draw, rounded corners, right=of theta, minimum width=2.5cm, minimum height=.8cm] (role) {visible role $r$};
\node[draw, rounded corners, right=of role, minimum width=2.5cm, minimum height=.8cm] (gate) {norm gate $N$};
\node[draw, rounded corners, right=of gate, minimum width=2.5cm, minimum height=.8cm] (trace) {trace $Y_{1:T}$};
\node[draw, rounded corners, below=of trace, minimum width=2.5cm, minimum height=.8cm] (test) {audience tests $\F_T$};
\node[draw, rounded corners, below=of gate, minimum width=2.5cm, minimum height=.8cm] (res) {residue};
\draw[->] (theta) -- (role);
\draw[->] (role) -- (gate);
\draw[->] (gate) -- (trace);
\draw[->] (trace) -- (test);
\draw[->] (test) -- (res);
\draw[->] (res) -- (gate);
\end{tikzpicture}
\caption{The staged public carrier: hidden locus, visible role, norm gate, public trace, audience test, and residual feedback.}
\label{fig:stage-map}
\end{figure}

The loop from residue back to the norm gate expresses an institutional fact: once a pattern becomes readable, future actors may change their movement, and the grammar of the room may change with them. The core theory in this paper fixes a horizon $T$ and treats the room as given. The loop is included because it explains why small residues matter.

\section{Couplings and chorus witnesses}

The first useful way to certify choric masking is through couplings. A coupling is a joint law with prescribed marginals~\cite{villani2009}. It lets us compare two traces inside one probability space.

\begin{definition}[Statistic-sufficient coupling]
Let $\F_T$ be a test class. A map $\sigma:Y^T\to Z$ is $\F_T$-sufficient if every $f\in\F_T$ factors as $f=g_f\circ \sigma$ for some $g_f:Z\to[0,1]$. A coupling $\Gamma$ of laws $P,Q$ on $Y^T$ is $\eta$-synchronizing for $\sigma$ if
\[
    \Gamma\{(y,z):\sigma(y)\neq\sigma(z)\}\leq \eta.
\]
\end{definition}

\begin{theorem}[Coupling criterion]
Let $P,Q$ be laws on $Y^T$. Suppose $\sigma$ is $\F_T$-sufficient and there exists an $\eta$-synchronizing coupling of $P$ and $Q$ for $\sigma$. Then
\[
    \Delta_{\F_T}(P,Q)\leq \eta.
\]
Consequently, if $P\in\operatorname{Ch}_T(\theta)$ and $Q\in\operatorname{Ch}_T(\theta')$, then $\theta$ and $\theta'$ are $(\F_T,\eta)$-chorically masked.
\end{theorem}

\begin{proof}
Let $f\in\F_T$. Since $\sigma$ is sufficient, $f=g\circ\sigma$ for some $g:Z\to[0,1]$. If $(Y,Z)$ has coupling $\Gamma$, then
\[
\begin{aligned}
\left|\mathbb E_P f-\mathbb E_Q f\right|
&=\left|\mathbb E_\Gamma[g(\sigma(Y))-g(\sigma(Z))]\right|\\
&\leq \mathbb E_\Gamma\left|g(\sigma(Y))-g(\sigma(Z))\right|.
\end{aligned}
\]
The integrand is at most $1$ and is zero whenever $\sigma(Y)=\sigma(Z)$. Therefore it is bounded by the indicator of the disagreement event. The expectation is at most $\eta$. Taking the supremum over $f\in\F_T$ proves the claim.\end{proof}

\begin{remark}
The theorem is useful when the audience does not see the full trace but only a public statistic: phase, count, rank, coarse time, or a bounded window. The room need not synchronize every gesture. It must synchronize what the audience can use.
\end{remark}

\begin{definition}[Chorus witness]
A chorus witness between $\theta$ and $\theta'$ consists of finite lists
\[
    (\alpha_i,r_i,\pi_i)_{i=1}^m,
    \qquad
    (\beta_j,r'_j,\pi'_j)_{j=1}^n,
\]
where $\alpha_i,\beta_j\geq0$, $\sum_i\alpha_i=\sum_j\beta_j=1$, $\pi_i\in\Pi_T(r_i)$, $\pi'_j\in\Pi_T(r'_j)$, and
\[
    P=\sum_{i=1}^m \alpha_i P_{\theta,r_i}^{\pi_i,T},
    \qquad
    Q=\sum_{j=1}^n \beta_j P_{\theta',r'_j}^{\pi'_j,T}
\]
satisfy $\Delta_{\F_T}(P,Q)\leq\eps$.
\end{definition}

A chorus witness is finite data. It can be audited by computing the trace laws, the mixture weights, and the audience distance.

\begin{theorem}[Convex chorus theorem]
For hidden loci $\theta,\theta'$ and horizon $T$, the following are equivalent.
\begin{enumerate}[label=(\alph*),leftmargin=2em]
\item $\theta$ and $\theta'$ are perfectly chorically masked for the full test class $[0,1]^{Y^T}$.
\item $\operatorname{Ch}_T(\theta)\cap\operatorname{Ch}_T(\theta')\neq\varnothing$.
\end{enumerate}
More generally, for any test class $\F_T$, they are $(\F_T,\eps)$-chorically masked if and only if
\[
    \inf_{P\in\operatorname{Ch}_T(\theta),\ Q\in\operatorname{Ch}_T(\theta')}
    \Delta_{\F_T}(P,Q)\leq\eps.
\]
\end{theorem}

\begin{proof}
For the full test class, $\Delta_{[0,1]^{Y^T}}(P,Q)=0$ if and only if $P=Q$. Therefore perfect masking is equivalent to the existence of a common law in the two choric hulls. The approximate statement is exactly the definition of $(\F_T,\eps)$-masking after taking the infimum over admissible choric laws.\end{proof}

\section{Residues are finite}

A public trace space at horizon $T$ is finite. This gives a strong statement: if there is leakage, some finite trace pattern carries it. The next theorem is the residue theorem.

\begin{theorem}[Residue theorem]
Let $P,Q$ be laws on finite $Y^T$. If $\TV(P,Q)>\eps$, then there exists an event $C\subseteq Y^T$ such that
\[
    |P(C)-Q(C)|>\eps.
\]
Moreover, $C$ can be chosen as
\[
    C^+=\{y\in Y^T:P(y)>Q(y)\}
\]
or as its complement. If $C$ is decomposed into singleton cylinders, at least one full cylinder $\{y\}$ satisfies
\[
    |P(y)-Q(y)|\geq \frac{2\,\TV(P,Q)}{|Y|^T}.
\]
\end{theorem}

\begin{proof}
For finite spaces,
\[
    \TV(P,Q)=\sup_{C\subseteq Y^T}|P(C)-Q(C)|
    =P(C^+)-Q(C^+).
\]
Thus $C^+$ is a maximizing event. Since
\[
    \TV(P,Q)=\frac12\sum_{y\in Y^T}|P(y)-Q(y)|,
\]
at least one point has absolute mass difference at least $2\TV(P,Q)/|Y|^T$, which is the displayed bound.\end{proof}

\begin{corollary}[Observable residue]
If $\Delta_{\F_T}(P,Q)>\eps$ for a test class containing indicators of a field of events $\C$, then there exists $C\in\C$ with
\[
    |\res_{P,Q}(C)|>\eps.
\]
\end{corollary}

\begin{proof}
By definition, some $f\in\F_T$ separates the expectations by more than $\eps$. If $\F_T$ contains only indicators of events in $\C$, this $f$ is the indicator of such an event.\end{proof}

\begin{remark}
For general bounded tests, a layer-cake argument converts the separating test into a threshold event. If $f:Y^T\to[0,1]$, then
\[
    \mathbb E_P f-\mathbb E_Q f
    =\int_0^1 \big(P\{f\geq t\}-Q\{f\geq t\}\big)\,dt.
\]
Therefore some threshold set carries at least the same signed separation.
\end{remark}

The residue theorem is the place where the room becomes concrete. If the audience can separate two loci, there is an event in the public trace space that does the separating. The event may be awkward to name, but it exists. In a finite room, exposure always has a carrier.

\section{Unique gestures and cover mass}

A gesture is unique when one locus can produce it with mass that other loci cannot match. Unique gestures are dangerous because they defeat choric masking unless the room supplies enough cover mass.

\begin{definition}[Gesture event]
A gesture event at horizon $T$ is any event $G\subseteq Y^T$ in a designated event family $\G_T$. A gesture family may contain cylinders, count events, timing bands, or phase patterns.
\end{definition}

\begin{definition}[Cover mass]
For hidden locus $\theta$, event $G$, and horizon $T$, define the attainable mass interval
\[
    I_T(\theta,G)=\left\{P(G):P\in\operatorname{Ch}_T(\theta)\right\}\subseteq [0,1].
\]
The cover gap between loci $\theta,\theta'$ on $G$ is
\[
    \gamma_T(\theta,\theta';G)=\operatorname{dist}\big(I_T(\theta,G),I_T(\theta',G)\big),
\]
where the distance between intervals is zero when they intersect and otherwise is the distance between their closest endpoints.
\end{definition}

\begin{theorem}[Unique-gesture lower bound]
Assume the audience test class $\F_T$ contains the indicator $1_G$ of a gesture event $G$. Then for every $P\in\operatorname{Ch}_T(\theta)$ and $Q\in\operatorname{Ch}_T(\theta')$,
\[
    \Delta_{\F_T}(P,Q)\geq |P(G)-Q(G)|.
\]
Consequently,
\[
    \inf_{P,Q}\Delta_{\F_T}(P,Q)
    \geq \gamma_T(\theta,\theta';G),
\]
where the infimum is over $P\in\operatorname{Ch}_T(\theta)$ and $Q\in\operatorname{Ch}_T(\theta')$.
\end{theorem}

\begin{proof}
Since $1_G\in\F_T$,
\[
    \Delta_{\F_T}(P,Q)
    =\sup_{f\in\F_T}|\mathbb E_P f-\mathbb E_Q f|
    \geq |\mathbb E_P1_G-\mathbb E_Q1_G|
    =|P(G)-Q(G)|.
\]
Taking the infimum over choric laws gives the distance between attainable intervals.\end{proof}

\begin{corollary}[No free disappearance]
If some gesture $G$ is mandatory for $\theta$ in the sense that $P(G)\geq a$ for all $P\in\operatorname{Ch}_T(\theta)$, while $P'(G)\leq b$ for all $P'\in\operatorname{Ch}_T(\theta')$, with $a>b$, then $\theta$ and $\theta'$ cannot be $(\F_T,\eps)$-chorically masked for any $\eps<a-b$.
\end{corollary}

\begin{proof}
The attainable intervals are contained in $[a,1]$ and $[0,b]$. Their distance is at least $a-b$. Apply the theorem.\end{proof}

This is the first impossibility principle. If the room forces one hidden locus to make a gesture that others do not make, the gesture is not neutral. It is a public mark. Randomness cannot remove a mark that the norm gate makes mandatory.

\section{Repeated windows}

A small residue may look harmless in one window. Repetition changes the scale. If a room is visited many times, a weak event can become a stable exposure channel.

\begin{definition}[Window product]
For a law $P$ on $Y^T$, let $P^{\otimes n}$ be the law of $n$ independent windows, each distributed as $P$. The product trace space is $(Y^T)^n$.
\end{definition}

\begin{theorem}[Product upper bound]
For laws $P,Q$ on $Y^T$,
\[
    \TV(P^{\otimes n},Q^{\otimes n})\leq n\,\TV(P,Q).
\]
\end{theorem}

\begin{proof}
Use a maximal coupling~\cite{levin2009} $(X_1,Z_1)$ of $P,Q$ with $\mathbb P\{X_1\neq Z_1\}=\TV(P,Q)$. Take $n$ independent copies. Then the product coupling disagrees in at least one coordinate with probability at most $n\TV(P,Q)$ by the union bound. Total variation is the smallest possible disagreement probability over couplings, so the result follows.\end{proof}

\begin{theorem}[Gesture amplification]
Let $G\subseteq Y^T$ and suppose $P(G)=p$, $Q(G)=q$. In $n$ independent windows, define
\[
    H_n=\{\text{the gesture }G\text{ occurs at least once}\}.
\]
Then
\[
    P^{\otimes n}(H_n)-Q^{\otimes n}(H_n)
    =(1-(1-p)^n)-(1-(1-q)^n).
\]
If $p>q$, this equals
\[
    (1-q)^n-(1-p)^n.
\]
For small $p,q$ and moderate $n$, the separation is approximately $n(p-q)$ as long as $np$ and $nq$ remain small.
\end{theorem}

\begin{proof}
The event $H_n^c$ is the event that $G$ fails in every window. Under $P^{\otimes n}$ this has probability $(1-p)^n$, and under $Q^{\otimes n}$ it has probability $(1-q)^n$. Subtracting gives the formula. The approximation follows from the first-order expansion $(1-x)^n=1-nx+O(n^2x^2)$.\end{proof}

\begin{remark}
The theorem explains why rare public gestures are unstable under repetition. A one-percent difference can remain invisible in one encounter and become obvious over a season of repeated encounters.
\end{remark}

\begin{proposition}[Cumulative risk is not the next-window hazard]
Let capture events $E_1,\ldots,E_n$ be conditionally independent with a fixed one-window probability $p$. Then
\[
    \Pr\{E_1\cup\cdots\cup E_n\}=1-(1-p)^n,
\]
but for every $k$ one still has $\Pr\{E_{k+1}\mid E_1^c,\ldots,E_k^c\}=p$. Thus repetition increases cumulative event probability without, by itself, increasing the next-window hazard. A model that claims increasing one-window hazard must therefore contain an adaptive state variable, such as debt, pressure, attention, or policy update.
\end{proposition}

\begin{proof}
The union formula is the complement of $n$ failures, each of probability $1-p$. Conditional independence gives
\[
\Pr\{E_{k+1}\mid E_1^c,\ldots,E_k^c\}=\Pr\{E_{k+1}\}=p.
\]
The final sentence follows because any change in the next-window hazard cannot be inferred from independence alone; it requires a changed law for the next window.
\end{proof}

\begin{proposition}[Adaptive hazard decomposition]
For an adaptive repeated room, write the carrier hazard as
\[
    h_{\sigma,k}=1-\exp(-\alpha_{\sigma,k}r_{\sigma,k}).
\]
Then an increase of the room pressure state can raise, lower, or leave unchanged the carrier hazard according to the sign of
\[
    \Delta(\alpha_{\sigma,k}r_{\sigma,k}).
\]
In particular, pressure growth alone does not imply carrier-risk growth; pressure must be coupled to carrier-local attention or carrier-local residue.
\end{proposition}

\begin{proof}
The map $x\mapsto1-e^{-x}$ is strictly increasing on $[0,\infty)$. Therefore the sign of the hazard change is the sign of the change in the product $\alpha_{\sigma,k}r_{\sigma,k}$. A pressure state affects the hazard only through variables that enter this product.
\end{proof}

\section{Repeated unresolved pressure}

The previous theorem treats repetition as a product experiment. A second phenomenon appears when the room itself reacts to unresolved residue. Repetition then does not merely amplify a statistic. It changes the pressure state of the room. The danger is not that a hidden locus has become invisible in an absolute sense. The danger is that its residue is read as a diffuse disturbance and the cost of that disturbance is charged to a cover population.

This section formalizes that failure mode as a security defect: response can broaden without localizing the source.

\begin{definition}[Candidate field and resolution hazard]
At repeated window $k$, let $\mathcal I_k$ be the finite set of visible candidates in the room. Each $i\in\mathcal I_k$ has raw residue $r_{i,k}\geq0$ and receives attention $\alpha_{i,k}\geq0$ with
\[
    \sum_{i\in\mathcal I_k}\alpha_{i,k}\leq B_k^{\rm att}.
\]
The one-window resolution hazard of $i$ is
\[
    h_{i,k}=1-\exp(-\alpha_{i,k}r_{i,k}).
\]
This is a reduced-form detection channel: residue without attention has zero hazard, and attention without residue has zero hazard.
\end{definition}

\begin{definition}[Unresolved residue and choric risk debt]
Let $u_k\geq0$ be the amount of residue injected into window $k$ that is not localized to its carrier by the end of that window. Let $c_k\geq0$ be localized clearance. The choric risk debt is the nonnegative recursion
\[
    D_{k+1}=\bigl[\lambda D_k+u_k-c_k\bigr]_+,
    \qquad 0\leq\lambda\leq1,
\]
where $[x]_+=\max\{x,0\}$. The parameter $\lambda$ is institutional memory. If $\lambda$ is close to one, unresolved pressure decays slowly.
\end{definition}

\begin{lemma}[Debt accounting identity]
If the positive part is inactive on windows $0,\ldots,n-1$, so that $D_{k+1}=\lambda D_k+u_k-c_k$, then
\[
    D_n=\lambda^nD_0+\sum_{k=0}^{n-1}\lambda^{n-1-k}(u_k-c_k).
\]
If the positive part is active, the same right-hand side is a lower affine ledger before reflection at zero and the realized debt is its reflected nonnegative version. Hence debt can decrease only through decay or localized clearance; unresolved residue cannot be made harmless by being merely unassigned.
\end{lemma}

\begin{proof}
Iterating the affine recursion gives the displayed identity. With the positive part, each step replaces a negative affine value by zero, which is reflection at the nonnegative boundary. The interpretation follows directly from the signs of the terms in the ledger.
\end{proof}

\begin{definition}[Displaced-risk choric masking]
Fix a carrier $\sigma_k\in\mathcal I_k$. Displaced-risk choric masking occurs on a sequence of windows if the following three inequalities hold over a nontrivial interval of $k$:
\[
    h_{\sigma_k,k}\leq \bar h,
    \qquad
    D_{k+1}-D_k\geq d>0,
    \qquad
    \sum_{j\neq \sigma_k}\Pr\{j\text{ is selected at }k\}
      \text{ increases with }D_k.
\]
Thus the carrier hazard remains bounded while unresolved pressure and non-carrier selection load rise. The residue has not disappeared. It has been displaced.
\end{definition}

\begin{lemma}[Debt lower envelope]
Suppose $D_{k+1}\geq\lambda D_k+\delta$ for $k=0,\ldots,n-1$, with $\delta>0$ and $0\leq\lambda\leq1$. Then
\[
D_n\geq
\begin{cases}
    \lambda^nD_0+\delta\frac{1-\lambda^n}{1-\lambda}, & 0\leq\lambda<1,\\[0.6em]
    D_0+n\delta, & \lambda=1.
\end{cases}
\]
\end{lemma}

\begin{proof}
Iterate the affine lower bound. For $0\leq\lambda<1$ the accumulated increments form a finite geometric sum. For $\lambda=1$ the increments add linearly.
\end{proof}

\begin{proposition}[Survivorship miscalibration]
Let a participant observe only a personal capture count $C_n$ over $n$ windows and form the smoothed empirical estimate
\[
    \widehat p_n=\frac{a+C_n}{a+b+n},
    \qquad a,b>0.
\]
If $C_n=0$ for all observed windows, then $\widehat p_n\to0$. At the same time, if unresolved residue satisfies the hypothesis of the debt lower envelope, then $D_n$ converges to a positive floor when $\lambda<1$ and diverges linearly when $\lambda=1$. Hence personal non-capture and system safety are not equivalent observables.
\end{proposition}

\begin{proof}
The estimate tends to zero because its numerator stays fixed while its denominator grows. The debt claim is exactly the preceding lemma. The two quantities are functions of different observations: one is a personal stopping history, the other is a room-level unresolved pressure state.
\end{proof}

\begin{definition}[Black-box breadth response]
A repeated room has a black-box breadth response if its selection breadth is a nondecreasing function
\[
    m_k=m(D_k),
\]
of the debt state, but the policy does not have a certificate that identifies the carrier inside the selected cover cell. The response is localization-blind on a cover cell $C_k\subseteq\mathcal I_k$ if all information available to the response policy is invariant under permutations of $C_k$.
\end{definition}

\begin{theorem}[Breadth-dilution lower bound]
Assume window $k$ contains a localization-blind cover cell $C_k$ of size $m_k$ and total attention budget $B_k^{\rm att}$. For any allocation of attention inside $C_k$, there exists a carrier placement $\sigma_k\in C_k$ with
\[
    \alpha_{\sigma_k,k}\leq \frac{B_k^{\rm att}}{m_k}.
\]
Consequently, if $r_{\sigma_k,k}\leq r_{\max}$, then
\[
    h_{\sigma_k,k}\leq
    1-\exp\left(-\frac{B_k^{\rm att}r_{\max}}{m_k}\right).
\]
If $m(D)$ increases while $B_k^{\rm att}$ and $r_{\max}$ are fixed, this upper bound decreases.
\end{theorem}

\begin{proof}
By the pigeonhole principle, among the $m_k$ members of $C_k$ at least one receives attention at most $B_k^{\rm att}/m_k$. In a localization-blind cell, the available evidence does not distinguish which member is the carrier, so a worst-case carrier placement may be at such a least-attended member. Substituting this attention bound into $h=1-\exp(-\alpha r)$ and using $r\leq r_{\max}$ gives the inequality. Monotonicity in $m_k$ is immediate.
\end{proof}

\begin{corollary}[More screening need not mean better localization]
A black-box response that increases $m(D_k)$ in reaction to unresolved debt can increase the number of selected candidates while decreasing the guaranteed per-carrier resolution hazard. Thus response breadth and localization quality are distinct variables.
\end{corollary}

\begin{proof}
Increasing $m(D_k)$ increases breadth by definition. The preceding theorem gives a decreasing upper bound on the worst-case carrier hazard when attention budget is fixed.
\end{proof}

\begin{proposition}[Softmax spillover derivative]
Let candidates have scores
\[
    S_i(D)=a_i+\gamma D q_i,
    \qquad \gamma>0,
\]
and selection probabilities
\[
    \pi_i(D)=\frac{\exp(\beta S_i(D))}{\sum_{\ell}\exp(\beta S_{\ell}(D))},
    \qquad \beta>0.
\]
Then
\[
    \frac{d}{dD}\pi_i(D)
    =\beta\gamma\pi_i(D)
      \left(q_i-\sum_{\ell}\pi_{\ell}(D)q_{\ell}\right).
\]
Hence every non-carrier whose pressure-coupling $q_i$ is above the current selected average receives increasing selection probability as debt rises.
\end{proposition}

\begin{proof}
Differentiate the normalized exponential expression. The derivative of the numerator is $\beta\gamma q_i$ times the numerator; the derivative of the log normalizer is $\beta\gamma\sum_\ell \pi_\ell(D)q_\ell$. Subtracting gives the formula.
\end{proof}

\begin{proposition}[False-positive load monotonicity]
In the softmax spillover model, suppose a non-carrier set $J$ satisfies
\[
    q_j\geq \sum_{\ell}\pi_{\ell}(D)q_{\ell}+a
    \qquad (j\in J)
\]
throughout an interval $D\in[D_0,D_1]$, for some $a>0$. Then
\[
    \frac{d}{dD}\sum_{j\in J}\pi_j(D)
    \geq \beta\gamma a\sum_{j\in J}\pi_j(D),
\]
and therefore
\[
    \sum_{j\in J}\pi_j(D_1)
    \geq
    \exp(\beta\gamma a(D_1-D_0))
    \sum_{j\in J}\pi_j(D_0)
\]
as long as the displayed separation condition persists.
\end{proposition}

\begin{proof}
Sum the spillover derivative over $j\in J$ and use the assumed lower bound on each parenthesized term. The exponential bound is Gronwall's inequality applied to the resulting differential inequality.
\end{proof}

\begin{theorem}[Pressure without localization]
Consider a repeated room satisfying the following conditions.
\begin{enumerate}[label=(\roman*),leftmargin=2.2em]
\item unresolved residue is persistently positive: $u_k-c_k\geq\delta>0$;
\item breadth is debt-responsive: $m_k=m(D_k)$ is nondecreasing;
\item the carrier remains inside a localization-blind cover cell of size at least $m_k$;
\item total attention is bounded by $B^{\rm att}$ and carrier residue is bounded by $r_{\max}$.
\end{enumerate}
Then $D_k$ has the lower envelope of the debt lemma, while the guaranteed carrier hazard satisfies
\[
    h_{\sigma_k,k}\leq
    1-\exp\left(-\frac{B^{\rm att}r_{\max}}{m(D_k)}\right).
\]
If $m(D)\to\infty$ along the debt trajectory, the guaranteed carrier hazard tends to zero even though room pressure increases.
\end{theorem}

\begin{proof}
Condition (i) gives $D_{k+1}\geq\lambda D_k+\delta$, hence the debt lower envelope. Conditions (ii)--(iv) allow the breadth-dilution theorem to be applied at every window with $m_k=m(D_k)$. If $m(D_k)$ diverges, the exponent $B^{\rm att}r_{\max}/m(D_k)$ tends to zero, and therefore the hazard upper bound tends to zero.
\end{proof}

\begin{theorem}[Displaced-risk debt certificate]
Assume the hypotheses of pressure without localization. Suppose also that a non-carrier subpopulation $J_k\subseteq\mathcal I_k\setminus\{\sigma_k\}$ has pressure-couplings $q_j$ above the selected average in the softmax spillover model. Then there is an interval of windows on which
\[
    D_k \text{ increases},
    \qquad
    h_{\sigma_k,k} \text{ is not guaranteed to increase},
    \qquad
    \sum_{j\in J_k}\pi_j(D_k) \text{ increases}.
\]
This is a certificate of displaced-risk choric masking.
\end{theorem}

\begin{proof}
Debt increase follows from the debt lower envelope. The absence of guaranteed carrier-hazard increase follows from the breadth-dilution bound, which can decrease as $m(D_k)$ grows. The spillover derivative is positive for every $j$ whose $q_j$ is above the current selected average, so the total selection probability of such a subpopulation increases on intervals where this strict inequality persists.
\end{proof}

\begin{theorem}[Budget required to defeat cover dilution]
Let a cover cell have size $m$, and suppose every possible carrier in the cell has residue at least $r_{\min}>0$. To guarantee one-window hazard at least $\zeta\in(0,1)$ for every carrier placement using only attention allocation, the total attention budget assigned to the cell must satisfy
\[
    B^{\rm att}\geq \frac{m}{r_{\min}}\log\frac{1}{1-\zeta}.
\]
\end{theorem}

\begin{proof}
Guaranteeing hazard $\zeta$ for every carrier placement requires every member of the cell to receive attention $\alpha_i$ satisfying
\[
    1-\exp(-\alpha_i r_{\min})\geq\zeta.
\]
Equivalently, $\alpha_i\geq r_{\min}^{-1}\log(1/(1-\zeta))$. Summing this lower bound over $m$ members yields the budget requirement.
\end{proof}

\begin{definition}[Localization kernel]
A response policy has localization kernel $\kappa_{i,k}\in[0,1]$ if the debt increment assigned to candidate $i$ at window $k$ is $\kappa_{i,k}D_k$, with $\sum_i\kappa_{i,k}\leq 1$. The kernel is sound at margin $\eta>0$ for a carrier $\sigma_k$ if
\[
    \kappa_{\sigma_k,k}\geq \eta
    \quad\text{whenever}\quad D_k>0.
\]
It is blind on a cover cell $C_k$ if $\kappa_{i,k}$ is invariant under permutations of $C_k$.
\end{definition}

\begin{theorem}[Black-box pressure cannot create carrier information]
Let $\Sigma$ be uniform on a cover cell $C_k$ of size $m_k$, and let the response observe a transcript $Z_k$ together with a debt state $D_k$ that is blind on the cell: $D_k$ is a function of statistics invariant under permutations of $C_k$, so that $\Sigma$ is independent of $D_k$ and remains uniform given $D_k$. Suppose that conditioned on $D_k$ the transcript has mutual information
\[
    \MI(\Sigma;Z_k\mid D_k)\leq I_k.
\]
Then every carrier-selection rule based on $(Z_k,D_k)$ has error probability at least
\[
    1-\frac{I_k+\log2}{\log m_k}.
\]
In particular, if the debt state is a black-box aggregate and does not increase $I_k$, then increasing $D_k$ cannot by itself produce a reliable carrier-local response.  It can only change breadth, thresholds, or selection load.
\end{theorem}

\begin{proof}
Blindness gives $\MI(\Sigma;D_k)=0$ and keeps $\Sigma$ uniform on the cell given $D_k=d$. Condition on a value of $D_k=d$. Since $d$ is known to the response, the relevant observation is $Z_k$ under the conditional channel from $\Sigma$ to $Z_k$.  Applying the finite localization lower bound conditionally gives
\[
    \Pr\{\widehat\Sigma\neq\Sigma\mid D_k=d\}
    \geq
    1-\frac{\MI(\Sigma;Z_k\mid D_k=d)+\log2}{\log m_k}.
\]
Averaging over $D_k$ and using the assumed conditional mutual-information bound gives the result.  A scalar aggregate can affect a policy, but unless it changes the information channel about which member of the cover cell is the carrier, it cannot overcome the Fano lower bound.
\end{proof}

\begin{lemma}[Symmetry forces diluted localization]
If a localization kernel is blind on a cover cell $C_k$ of size $m_k$ and assigns total mass at most one inside that cell, then some candidate in $C_k$ has $\kappa_{i,k}\leq 1/m_k$. If the carrier cannot be distinguished inside $C_k$, no policy using only this blind kernel can guarantee $\kappa_{\sigma_k,k}>1/m_k$ for all carrier placements.
\end{lemma}

\begin{proof}
The first claim is the pigeonhole principle applied to $\sum_{i\in C_k}\kappa_{i,k}\leq1$. For the second claim, permutation invariance means that the available response information is identical after relabeling members of $C_k$. A worst-case carrier placement can therefore choose a member whose kernel mass is at most the average.
\end{proof}

\begin{theorem}[Localization completeness is necessary for pressure to become hazard]
Assume the one-window hazard of the carrier has the form
\[
    h_{\sigma,k}=1-\exp(-\alpha_{\sigma,k}r_{\sigma,k})
\]
where the policy converts kernel mass into attention proportionally, so that $\alpha_{i,k}=B^{\rm att}\kappa_{i,k}$, and $r_{\sigma,k}\leq r_{\max}$. If the only response to debt is blind on a cover cell of size $m_k$, then the worst-case guarantee of any blind policy satisfies
\[
    \sup_{\text{blind policies}}\inf_{\sigma_k\in C_k} h_{\sigma,k}
    \leq
    1-\exp\left(-\frac{B^{\rm att}r_{\max}}{m_k}\right).
\]
In contrast, if the policy has a sound localization kernel with margin $\eta$ and the carrier residue is at least $r_{\min}>0$, then
\[
    h_{\sigma,k}\geq 1-\exp(-B^{\rm att}\eta r_{\min}).
\]
Thus pressure improves source-specific hazard only after it is converted into localized attention or localized residue.
\end{theorem}

\begin{proof}
Under blind response, the symmetry lemma gives a carrier placement with $\kappa_{\sigma,k}\leq1/m_k$, hence $\alpha_{\sigma,k}=B^{\rm att}\kappa_{\sigma,k}\leq B^{\rm att}/m_k$. Substituting $r_{\sigma,k}\leq r_{\max}$ into the increasing hazard map gives the upper bound. Under a sound localization kernel, $\kappa_{\sigma,k}\geq\eta$, so $\alpha_{\sigma,k}=B^{\rm att}\kappa_{\sigma,k}\geq B^{\rm att}\eta$. With $r_{\sigma,k}\geq r_{\min}$, substitution gives the lower bound.
\end{proof}

\begin{definition}[Externalized load]
Let $\sigma_k$ be the carrier and let $\pi_i(D_k)$ be the probability that candidate $i$ is selected at debt level $D_k$. The externalized load is
\[
    E_k(D_k)=\sum_{i\in\mathcal I_k\setminus\{\sigma_k\}}\pi_i(D_k).
\]
A repeated room externalizes risk on an interval if $D_k$ increases on that interval while $E_k(D_k)$ increases and the guaranteed carrier hazard does not increase.
\end{definition}

\begin{theorem}[Externalization trichotomy]
In a repeated room with persistent unresolved residue, at least one of the following must occur on any interval where debt increases.
\begin{enumerate}[label=(\roman*),leftmargin=2.2em]
\item Debt is cleared by localized resolution, so $c_k$ eventually matches the incoming unresolved residue.
\item Debt is converted into carrier-local attention or carrier-local residue, so the carrier hazard obtains a positive lower bound of the sound-localization form.
\item The response remains localization-blind on a growing cover cell, in which case externalized load can increase while the minimax carrier-hazard guarantee is bounded by the dilution term.
\end{enumerate}
\end{theorem}

\begin{proof}
If localized clearance catches incoming unresolved residue, the debt recursion stops increasing except for memory effects, giving (i). If debt does not clear but the response supplies a sound localization kernel or an equivalent carrier-local residue increase, the preceding theorem gives (ii). If neither clearance nor localization occurs, the response information remains invariant on the unresolved cover cell. A debt-responsive policy can only broaden or reweight that cell without breaking symmetry. The breadth-dilution theorem and softmax spillover derivative then give (iii): non-carrier selection can rise while the guaranteed carrier hazard is controlled by the inverse cover size.
\end{proof}

\begin{corollary}[Repeated non-capture is not a safety certificate]
Let $C_n$ be the event that the carrier has not been resolved in the first $n$ windows. If the room satisfies the third branch of the externalization trichotomy, then observing $C_n$ is compatible with increasing debt and increasing non-carrier load. Therefore a history of non-capture cannot certify absence of residue unless accompanied by a localized-clearance certificate or a valid upper bound on incoming unresolved residue.
\end{corollary}

\begin{proof}
Under the third branch, debt can grow by the debt lower envelope while the dilution bound prevents a guaranteed increase of carrier hazard. The same branch allows non-carrier load to increase by the spillover derivative. Hence the survival event $C_n$ only states that no resolution occurred; it does not identify the reason. It is consistent with no residue, with missed residue, or with displaced residue. A safety certificate must rule out the latter two cases by clearance or by an upper bound on $u_k$.
\end{proof}

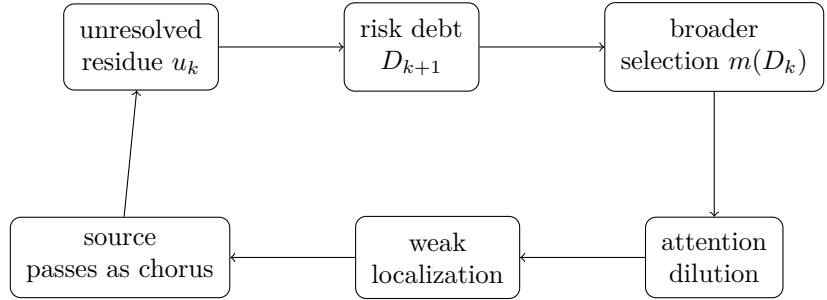
\begin{figure}[H]
\centering
\begin{tikzpicture}[node distance=1.65cm, every node/.style={draw, rounded corners, align=center, font=\small, inner sep=6pt}]
\node (r) {unresolved\\residue $u_k$};
\node[right=of r] (d) {risk debt\\$D_{k+1}$};
\node[right=of d] (m) {broader\\selection $m(D_k)$};
\node[below=of m] (a) {attention\\dilution};
\node[left=of a] (h) {weak\\localization};
\node[left=of h] (s) {source\\passes as chorus};
\draw[->] (r) -- (d);
\draw[->] (d) -- (m);
\draw[->] (m) -- (a);
\draw[->] (a) -- (h);
\draw[->] (h) -- (s);
\draw[->] (s) -- (r);
\end{tikzpicture}
\caption{A debt loop. The defect is not broad response itself; the defect is broad response without localization.}
\label{fig:debt-loop}
\end{figure}

The repeated-room variables should be read as an accounting chain rather than a list. The quantity $u_k$ is unresolved residue entering window $k$, while $c_k$ is localized clearance, meaning carrier-specific resolution rather than a generic reduction of anxiety. Their reflected difference feeds the debt state $D_k$. A response function $m(D_k)$ converts debt into selection breadth. When breadth grows faster than attention budget, the scale $B^{\rm att}/m(D_k)$ falls and gaze dilutes. Finally, the non-carrier load $\sum_{j\neq\sigma_k}\pi_j(D_k)$ records how much of the debt is being charged to candidates that did not generate the residue. The failure mode is not any one variable increasing; it is the chain $u_k-c_k\to D_k\to m(D_k)\to B^{\rm att}/m(D_k)\to$ non-carrier load without a matching localization certificate.

The mature lesson is narrow. Repetition does not by itself prove that the next window is more dangerous for the same carrier. If windows are independent and the hazard is fixed at $p$, the cumulative event probability is $1-(1-p)^n$. If the room is adaptive, the one-window hazard may change. But in a choric room, adaptation can be misdirected: pressure can be real, selection can broaden, and the carrier hazard can remain nearly unchanged because the response has not become more local.

\section{Norm gates and quotient visibility}

A norm gate can protect or expose. It protects when it permits enough common movement. It exposes when it forces one hidden locus into a narrower corridor.

\begin{definition}[Visible quotient]
Let $\nu:Y^T\to Z$ be a public quotient map. The pushforward of $P$ by $\nu$ is denoted $\nu_\#P$. A test class $\F_T$ is carried by $\nu$ if every $f\in\F_T$ factors through $\nu$.
\end{definition}

\begin{theorem}[Quotient masking]
Suppose $\F_T$ is carried by $\nu:Y^T\to Z$. If there exist $P\in\operatorname{Ch}_T(\theta)$ and $Q\in\operatorname{Ch}_T(\theta')$ such that
\[
    \nu_\#P=\nu_\#Q,
\]
then $\theta$ and $\theta'$ are perfectly chorically masked for $\F_T$.
\end{theorem}

\begin{proof}
For every $f\in\F_T$, write $f=g\circ\nu$. Then
\[
    \mathbb E_P f=\mathbb E_{\nu_\#P}g=\mathbb E_{\nu_\#Q}g=\mathbb E_Q f.
\]
The supremum of zero differences is zero.\end{proof}

\begin{definition}[Gate slack]
For role $r$, history $h$, and two hidden loci $\theta,\theta'$, define the local common action set
\[
    A_N(r,h)=\{a\in A:(r,h,a)\in N\}.
\]
The gate slack at $(r,h)$ is
\[
    \lambda_N(r,h)=\frac{|A_N(r,h)|}{|A|}.
\]
A room has slack floor $\lambda$ up to horizon $T$ if $\lambda_N(r,h)\geq\lambda$ for every role $r$ and visible history $h\in Y^{<T}$ reachable under some admissible script.
\end{definition}

Gate slack is a crude measure. It does not know which actions are useful. Still, a zero slack history is fatal: after such a history, the role has no admissible continuation.

\begin{proposition}[Support obstruction]
If for some event $C\subseteq Y^T$ all laws in $\operatorname{Ch}_T(\theta)$ assign zero mass to $C$, while some law in $\operatorname{Ch}_T(\theta')$ assigns mass at least $c>0$ to $C$, then the two choric hulls are disjoint. If $1_C\in\F_T$, their best possible $\F_T$-distance is at least the distance from $0$ to the attainable interval $I_T(\theta',C)$.
\end{proposition}

\begin{proof}
The first statement follows because a common law would have to assign both zero and positive mass to $C$. The second statement is the unique-gesture lower bound.\end{proof}

\section{Linked rooms}

Many public systems are not single rooms. A disclosure in one place changes attention in another. A release gate may delay a trace, coarsen it, or attach a public label. We need a composition rule.

\begin{definition}[Markov link]
Let $Y^T$ and $Z^U$ be trace spaces. A Markov link is a stochastic kernel
\[
    L:Y^T\to\Delta(Z^U).
\]
For a law $P$ on $Y^T$, write $PL$ for the induced law on $Z^U$.
\end{definition}

\begin{lemma}[Data contraction]
For every Markov link $L$,
\[
    \TV(PL,QL)\leq \TV(P,Q).
\]
More generally, for a test class $\G_U$ on $Z^U$, define the pulled-back class
\[
    L\G_U=\{y\mapsto \mathbb E_{z\sim L(y)}g(z):g\in\G_U\}.
\]
Then
\[
    \Delta_{\G_U}(PL,QL)=\Delta_{L\G_U}(P,Q).
\]
\end{lemma}

\begin{proof}
For the equality, compute
\[
    \mathbb E_{PL}g-\mathbb E_{QL}g
    =\mathbb E_P[Lg]-\mathbb E_Q[Lg],
\]
where $Lg(y)=\mathbb E_{z\sim L(y)}g(z)$. Taking suprema gives the second statement. For the first, take the full test class on $Z^U$. Its pullback is contained in the full bounded test class on $Y^T$, so total variation cannot increase.\end{proof}

\begin{theorem}[Linked-room accounting]
Let $P,Q$ be laws on the first room and let $L,L'$ be two Markov links into a second room. Then
\[
    \TV(PL,QL')\leq \TV(P,Q)+\sup_{y\in Y^T}\TV(L(y),L'(y)).
\]
\end{theorem}

\begin{proof}
Insert $QL$ as an intermediate law:
\[
    \TV(PL,QL')\leq \TV(PL,QL)+\TV(QL,QL').
\]
The first term is at most $\TV(P,Q)$ by contraction. For the second term, use convexity:
\[
\begin{aligned}
\TV(QL,QL')
&=\TV\left(\sum_y Q(y)L(y),\sum_y Q(y)L'(y)\right)\\
&\leq \sum_y Q(y)\TV(L(y),L'(y))\\
&\leq \sup_y\TV(L(y),L'(y)).
\end{aligned}
\]
Combining the bounds proves the result.\end{proof}

The accounting rule separates two forms of exposure: the first-room residue $\TV(P,Q)$ and the link discrepancy $\sup_y\TV(L(y),L'(y))$. A downstream public system can preserve masking if its link is common. It can also destroy masking by attaching different release behavior to otherwise similar traces.

\section{Post-release burden and residual object pressure}

The previous linked-room accounting separates first-room residue from link discrepancy. It still treats the link as an externally chosen stochastic channel. In many public rooms the second link is not arbitrary: it is forced by the state of an object, resource, obligation, or hazard that crossed the first gate with the bearer. The first room asks whether the trace can pass. The next room asks what the carried state does to the bearer's future scripts.

This section makes that pressure explicit. The point is structural: a gate should not be evaluated only by the trace visible at the gate when some hidden states create new public burdens immediately after release.

\begin{definition}[Object state]
An object state space is a finite set $X$ with coordinates written, when convenient, as
\[
    x=(q,\tau,b,c,h,\mu).
\]
Here $q$ is a quality coordinate, $\tau$ is a remaining-use window, $b$ is local availability, $c$ is replacement cost, $h$ is harm or spoilage risk, and $\mu$ is ambient trace intensity. These names are mnemonic only. Formally, $X$ is finite and each coordinate is a bounded statistic on $X$.
\end{definition}

\begin{definition}[Pressure functional]
Fix nonnegative weights $w=(w_\tau,w_h,w_\mu,w_c,w_b)$ and a constant $\eta>0$. The residual object pressure of $x\in X$ is
\[
    \Psi(x)
    =w_\tau\frac{1}{\tau(x)+\eta}
     +w_h h(x)
     +w_\mu \mu(x)
     +w_c c(x)
     +w_b(1-b(x)).
\]
A coordinate may be omitted by setting its weight to zero. A larger value of $\Psi$ means that the object state imposes a tighter future burden on admissible public scripts.
\end{definition}

\begin{definition}[Burdened post-release room]
A burdened post-release room over horizon $U$ consists of
\[
    \B=(Z,U,X,M,\Pi,\G_U),
\]
where $Z^U$ is the second-room trace space, $M:X\to\Delta(X)$ is an object-state transition kernel, $\G_U\subseteq[0,1]^{Z^U}$ is an audience class, and
\[
    \Pi:X\to 2^{Z^U}
\]
assigns to each object state the set of downstream scripts compatible with that burden. A law $P$ on $Z^U$ is compatible with $x$ when $P(\Pi(x))=1$.
\end{definition}

\begin{assumption}[Nested pressure gate]
The burden gate is pressure-monotone if
\[
    \Psi(x')\geq \Psi(x)
    \quad\Longrightarrow\quad
    \Pi(x')\subseteq \Pi(x).
\]
Thus pressure never creates more future scripts. It can only preserve or remove them.
\end{assumption}

\begin{definition}[Burden link]
Let $Y^T$ be the first-room trace space. A burden link is a stochastic kernel
\[
    L_X:Y^T\times X\to\Delta(Z^U)
\]
such that $L_X(y,x)$ is supported on $\Pi(x)$ for every $y$ and $x$. If $\nu$ is a law on $Y^T\times X$, write $\nu L_X$ for the induced second-room law on $Z^U$.
\end{definition}

The burden link is a formal version of the sentence: the object is not merely carried through the gate; it selects which future scripts remain possible.

\begin{lemma}[Pressure shrinks choric hulls]
Let $\operatorname{Ch}_U(x)$ be the convex hull of all second-room laws supported on $\Pi(x)$. Under the nested pressure gate, if $\Psi(x')\geq\Psi(x)$, then
\[
    \operatorname{Ch}_U(x')\subseteq \operatorname{Ch}_U(x).
\]
\end{lemma}

\begin{proof}
By pressure monotonicity, every script in $\Pi(x')$ is also a script in $\Pi(x)$. Therefore every law supported on $\Pi(x')$ is supported on $\Pi(x)$. Taking convex hulls preserves inclusion.
\end{proof}

\begin{definition}[Post-release burden]
For a law $\nu$ on $Y^T\times X$, the expected post-release burden is
\[
    \mathsf B(\nu)=\mathbb E_\nu \Psi(X).
\]
For two hidden loci with joint first-room/object laws $\nu_0,\nu_1$, the burden gap is
\[
    \operatorname{bgap}(\nu_1,\nu_0)=\left|\mathsf B(\nu_1)-\mathsf B(\nu_0)\right|.
\]
The burden gap is not itself a visible leakage statement. It becomes visible only through a burden link and a downstream test.
\end{definition}

\begin{lemma}[Reverse linked-room accounting]
Let $P,Q$ be first-room laws and let $L,L'$ be links into $Z^U$. For any symmetric test class $\G_U$,
\[
    \Delta_{\G_U}(PL,QL')
    \geq
    \Delta_{\G_U}(QL,QL')-\Delta_{\G_U}(PL,QL).
\]
In particular,
\[
    \Delta_{\G_U}(PL,QL')
    \geq
    \Delta_{\G_U}(QL,QL')-\TV(P,Q).
\]
\end{lemma}

\begin{proof}
The first inequality is the reverse triangle inequality for the seminorm $\Delta_{\G_U}$. For the second inequality, use data contraction to bound $\Delta_{\G_U}(PL,QL)$ by $\TV(PL,QL)\leq\TV(P,Q)$.
\end{proof}

\begin{theorem}[Gate passage is not closure]
Suppose two loci have first-room laws $P,Q$ with
\[
    \Delta_{\F_T}(P,Q)\leq\eps.
\]
Let $x_0,x_1\in X$ be post-release object states and let $L_0,L_1$ be burden links supported on $\Pi(x_0)$ and $\Pi(x_1)$. If there exists $g\in\G_U$ such that
\[
    \inf_{z\in\Pi(x_1)} g(z)-\sup_{z\in\Pi(x_0)} g(z)\geq \gamma>0,
\]
then every pair of downstream laws $P L_1$ and $Q L_0$ satisfies
\[
    \Delta_{\G_U}(P L_1,Q L_0)\geq \gamma.
\]
Thus first-room masking does not imply downstream masking when the carried state induces a separating burden gate.
\end{theorem}

\begin{proof}
Because $L_1$ is supported on $\Pi(x_1)$, every law $P L_1$ gives expectation at least $\inf_{z\in\Pi(x_1)}g(z)$. Because $L_0$ is supported on $\Pi(x_0)$, every law $Q L_0$ gives expectation at most $\sup_{z\in\Pi(x_0)}g(z)$. The difference of expectations is at least $\gamma$, so the supremum over $\G_U$ is at least $\gamma$.
\end{proof}

\begin{corollary}[Pressure destroys a chorus]
Assume the first room supplies a common choric law for two loci, so that their first-room distance can be zero. If their post-release object states $x_0,x_1$ satisfy the separating burden-gate condition in the preceding theorem, then the second room has exposure floor at least $\gamma$ even though the first room had a perfect mask.
\end{corollary}

\begin{proof}
Apply the theorem with $\eps=0$. The lower bound depends only on the downstream supports created by the object states, not on the first-room distance.
\end{proof}

\begin{definition}[Micro-residue burden score]
Let $\C_U$ be a family of downstream cylinder events. For $C\in\C_U$ define
\[
    A_x(C)=
    \frac{
      \delta_x(C)\,u_x(C)\,\rho_x(C)\,\ell_x(C)
    }{k(C)+\eta},
\]
where $\delta_x(C)$ is the raw separating residue induced by state $x$, $u_x(C)$ is its urgency, $\rho_x(C)$ is perceptual salience, $\ell_x(C)$ is localizability to the relevant layer, and $k(C)$ is attention cost. The score is a certificate statistic, not a behavioral recommendation.
\end{definition}

\begin{definition}[Attention-pressure certificate]
Let $C\in\C_U$ be a downstream cylinder and let $p_L(C)$ be the hit probability of an attention lens $L$ on $C$. The certified exposure product is
\[
    \mathfrak E_L(C;x)
    =\delta_x(C)\,p_L(C)\,u_x(C)\,\ell_x(C).
\]
A cylinder is $(\alpha,\beta,\chi)$-certified if
\[
    \delta_x(C)\geq\alpha,
    \qquad
    p_L(C)u_x(C)\ell_x(C)\geq\beta,
    \qquad
    \mathfrak E_L(C;x)\geq\chi.
\]
\end{definition}

\begin{proposition}[Three ways a residue fails]
For a fixed downstream cylinder $C$, failure of an attention-pressure certificate has exactly three formal sources:
\[
    \delta_x(C)<\alpha,
    \qquad
    p_L(C)u_x(C)\ell_x(C)<\beta,
    \qquad
    \delta_x(C)p_L(C)u_x(C)\ell_x(C)<\chi.
\]
They correspond respectively to no sufficient raw residue, sufficient residue without enough perceptual access, and sufficient local signal whose product with burden and linkage remains below the reporting threshold.
\end{proposition}

\begin{proof}
The certificate is the conjunction of the three displayed inequalities. Its negation is the disjunction of their failures. The interpretation follows from the factors in the definition.
\end{proof}

\begin{theorem}[Attention can hide what pressure later reveals]
Let $\theta_0,\theta_1$ be two loci. Suppose there are first-room laws $P_0,P_1$ and an attention lens $L^{\rm att}$ such that
\[
    \Delta_{\widetilde\F_T}(L^{\rm att}P_0,L^{\rm att}P_1)\leq\eps.
\]
Suppose also that their post-release states $x_0,x_1$ and burden links $L_0,L_1$ satisfy the separating burden-gate condition with margin $\gamma$. Then the same pair is attention-masked in the first room and exposed in the second room:
\[
    \Delta_{\widetilde\F_T}(L^{\rm att}P_0,L^{\rm att}P_1)\leq\eps,
    \qquad
    \Delta_{\G_U}(P_1L_1,P_0L_0)
    \geq\gamma.
\]
\end{theorem}

\begin{proof}
The first inequality is the hypothesis. The second is the gate-passage theorem. The two claims concern different rooms and different observation channels. There is no contradiction because the first bound is an attention-filtered statement about the entrance trace, while the second is a support-separation statement about downstream scripts selected by object pressure.
\end{proof}

\begin{example}[A finite burden chamber]
Take six cylinder events $C_1,\ldots,C_6$. The first four occur at the entrance room; the last two occur only after release. Let the burden score be
\[
    A(C)=\frac{\delta_C u_C\rho_C\ell_C}{k_C}.
\]
\Cref{tab:ledger} instantiates the score on the six events; the values are declared analytic parameters of the example, chosen to display one admissible certificate pattern.

\begin{table}[H]
\centering
\caption{An analytic burden ledger on six cylinder events. $C_1$--$C_4$ occur at the entrance; $C_5,C_6$ occur only after release.}
\label{tab:ledger}
\begin{tabular}{lcccccc}
\toprule
 & $C_1$ & $C_2$ & $C_3$ & $C_4$ & $C_5$ & $C_6$\\
\midrule
stage & entrance & entrance & entrance & entrance & post-release & post-release\\
$A(C)$ & .115 & .079 & .012 & .029 & .798 & .305\\
\bottomrule
\end{tabular}
\end{table}

The pattern is the proof obligation in miniature: the entrance events carry nonzero residue but low localizability or low aperture hit probability, so the first room remains masked, while the post-release events combine urgency and localizability and satisfy a downstream certificate. The same hidden state is masked at the gate and separated after it.
\end{example}

\begin{proposition}[Burden threshold separates script sets]
Let $\Pi_{\leq p}$ be the set of downstream scripts compatible with pressure at most $p$, and let $\Pi_{>p}$ be the set compatible with pressure greater than $p$. If an event $G\subseteq Z^U$ satisfies
\[
    \inf_{z\in\Pi_{>p}}1_G(z)-\sup_{z\in\Pi_{\leq p}}1_G(z)\geq\gamma,
\]
then $G$ is a downstream certificate for the pressure threshold, and any test class containing $1_G$ separates the two script families by at least $\gamma$.
\end{proposition}

\begin{proof}
This is the burden-gate theorem applied to the two support sets $\Pi_{>p}$ and $\Pi_{\leq p}$. Since $1_G$ is binary, the displayed inequality can only have margins in $[0,1]$, and a positive margin means all high-pressure scripts put more mass on $G$ than all low-pressure scripts.
\end{proof}

The resulting ledger is
\[
\boxed{
\text{visible exposure}
\quad\text{requires}\quad
\text{residue}\times\text{attention}\times\text{pressure}\times\text{room linkage}.
}
\]
This is not a numerical identity. It is a certificate grammar. A room can fail because the raw residue is absent, because the observer misses it, because the object has not yet exerted pressure, or because the downstream link does not carry the pressure into an observable script.

\section{A four-role chamber}

We now build a small finite room. The purpose is not realism. The purpose is to make the residue visible.

There are two hidden loci:
\[
    \Theta=\{\alpha,\beta\}.
\]
There are four visible roles:
\[
    R=\{\mathrm{guest},\mathrm{regular},\mathrm{custodian},\mathrm{broker}\}.
\]
The observation alphabet is
\[
Y=\{e,w,a,c,s,h\},
\]
where the symbols mean entry, wait, acknowledge, correction, silence, and handoff. The audience sees only these symbols. The horizon is $T=3$.

The action alphabet is
\[
A=\{\mathrm{enter},\mathrm{wait},\mathrm{ack},\mathrm{correct},\mathrm{silent},\mathrm{handoff}\}.
\]
For simplicity, take $S=\{0,1,2,3\}$ and let the state be the clock, so $S_t=t$. Thus $K(t+1\mid t,a,\theta)=1$ for $t<3$.

The norm gate is role-dependent and history-independent in this chamber; \Cref{tab:gate} lists the admissible actions. The gate is intentionally simple: exposure comes from the different action corridors made visible by ordinary role permissions.

\begin{table}[H]
\centering
\caption{Role gate of the four-role chamber. A dash means the action is not admissible for the role; the gate does not depend on the history.}
\label{tab:gate}
\begin{tabular}{lcccccc}
\toprule
role & enter & wait & ack & correct & silent & handoff\\
\midrule
guest     & $\checkmark$ & $\checkmark$ & $\checkmark$ & --           & $\checkmark$ & --\\
regular   & $\checkmark$ & $\checkmark$ & $\checkmark$ & $\checkmark$ & $\checkmark$ & --\\
custodian & $\checkmark$ & $\checkmark$ & $\checkmark$ & $\checkmark$ & $\checkmark$ & $\checkmark$\\
broker    & $\checkmark$ & $\checkmark$ & $\checkmark$ & --           & $\checkmark$ & $\checkmark$\\
\bottomrule
\end{tabular}
\end{table}

The observation kernel is nearly deterministic: the action usually emits its matching symbol. The hidden locus changes only two small tendencies. Locus $\alpha$ emits correction more reliably when correction is available; locus $\beta$ emits handoff more reliably when handoff is available. Formally, when action $a$ has matching symbol $y(a)$,
\[
    O(y(a)\mid t,a,r,\theta)=1-\delta_{a,\theta},
\]
with the remaining mass $\delta_{a,\theta}$ placed on silence $s$; the silent action emits $s$ with probability one. The deviations are
\[
\begin{array}{c|cc}
 & \theta=\alpha & \theta=\beta\\
\hline
 a=\mathrm{correct} & 0.02 & 0.10\\
 a=\mathrm{handoff} & 0.10 & 0.02\\
\end{array}
\]
and all other non-silent actions have deviation $0.04$ for both loci.

The chamber fixes one \emph{ceremonial script} per role: at every step the role draws uniformly among its admissible actions. Pinning the script class is part of the room --- etiquette acts as a gate at the script level --- so the choric hull is generated by role mixtures alone. Since the clock is the state and the gate is history-independent, the three observations are i.i.d.\ conditional on $(\theta,r)$, and \Cref{tab:onestep} lists the one-step laws.

\begin{table}[H]
\centering
\caption{One-step public laws under the ceremonial scripts, rounded to three decimals. Rows are role--locus pairs; columns are emitted symbols.}
\label{tab:onestep}
\begin{tabular}{llcccccc}
\toprule
role & locus & $e$ & $w$ & $a$ & $c$ & $s$ & $h$\\
\midrule
guest     & $\alpha$ & .240 & .240 & .240 & 0    & .280 & 0\\
guest     & $\beta$  & .240 & .240 & .240 & 0    & .280 & 0\\
regular   & $\alpha$ & .192 & .192 & .192 & .196 & .228 & 0\\
regular   & $\beta$  & .192 & .192 & .192 & .180 & .244 & 0\\
custodian & $\alpha$ & .160 & .160 & .160 & .163 & .207 & .150\\
custodian & $\beta$  & .160 & .160 & .160 & .150 & .207 & .163\\
broker    & $\alpha$ & .192 & .192 & .192 & 0    & .244 & .180\\
broker    & $\beta$  & .192 & .192 & .192 & 0    & .228 & .196\\
\bottomrule
\end{tabular}
\end{table}

These laws induce exact attainable one-step intervalsThe table induces exact attainable one-step intervals for the two symbols that matter most in the chamber:
\[
\begin{aligned}
I_1(\alpha,c)&=[0,.196], & I_1(\beta,c)&=[0,.180],\\
I_1(\alpha,h)&=[0,.180], & I_1(\beta,h)&=[0,.196].
\end{aligned}
\]
At one step the intervals overlap because the guest role supplies zero mass and the custodian/broker roles supply cover. If the ceremony forbids the guest role and requires a handoff-capable role, the intervals change to
\[
\begin{aligned}
I^{\mathrm{hg}}_1(\alpha,c)&=[0,.163], & I^{\mathrm{hg}}_1(\beta,c)&=[0,.150],\\
I^{\mathrm{hg}}_1(\alpha,h)&=[.150,.180], & I^{\mathrm{hg}}_1(\beta,h)&=[.163,.196].
\end{aligned}
\]
The correction coordinate still overlaps, but the handoff coordinate now has a possible floor mismatch: choosing broker-like cover for $\alpha$ cannot imitate custodian-like cover for $\beta$ without changing the public role grammar. This is the chamber's first sandbar: a small region where the trace law looks ordinary locally while the admissible hull has already narrowed.

\begin{remark}[Slack restores the chorus]
The sandbar exists exactly because the ceremony pins the scripts. If the script-level gate is relaxed so that any $N$-admissible script may be used, a handoff-capable role may simply never hand off, both handoff intervals regain the point $0$, and the floors disappear. The exposure in this chamber is therefore a property of the script-level gate, not of the roles alone --- a quantitative instance of gate slack and of design Rule~3 below: a room must state which script freedoms its participants actually have before its safety can be assessed.
\end{remark}

The guest role is a perfect local cover: it erases the difference between $\alpha$ and $\beta$ because neither correction nor handoff is available. Yet the guest role may not always be acceptable. If the public room requires at least one non-guest role in a release ceremony, the cover disappears.

Let the audience test family contain only counts of $c$ and $h$ over the three-step trace. Define
\[
    C_c=\{y_{1:3}: \text{at least one }c\text{ occurs}\},\qquad
    C_h=\{y_{1:3}: \text{at least one }h\text{ occurs}\}.
\]
For a one-step probability $p_c$, the event $C_c$ has probability $1-(1-p_c)^3$. For a one-step probability $p_h$, the event $C_h$ has probability $1-(1-p_h)^3$.

In the custodian role, the correction event has probabilities
\[
    1-(1-.163)^3\approx .414,
    \qquad
    1-(1-.150)^3\approx .386.
\]
The handoff event has probabilities
\[
    1-(1-.150)^3\approx .386,
    \qquad
    1-(1-.163)^3\approx .414.
\]
Each difference is about $.028$. Neither difference is dramatic in one chamber visit. But across $n=20$ independent visits, the at-least-once amplification for a one-window difference from $.386$ to $.414$ gives
\[
    (1-.386)^{20}-(1-.414)^{20}\approx 3.5\times 10^{-5}.
\]
That particular at-least-once event saturates too quickly because both base probabilities are large. A better audience test counts frequency. For binomial count tests, the separation appears around the mean. The standard deviation scale is of order $\sqrt{n}$, while the mean gap is of order $n$. Thus repeated visits convert a small one-step asymmetry into a stable frequency residue.

This chamber shows a simple lesson. The dangerous symbol is not the rare one. The dangerous symbol is the one whose attainable frequency interval fails to overlap after the room imposes its norm gate.

\section{Balanced chorus construction}

The previous chamber had a trivial perfect cover through the guest role. We now give the general construction behind such covers.

\begin{definition}[Common public kernel]
A probability law $W$ on $Y^T$ is a common public kernel for a set $B\subseteq\Theta$ if
\[
    W\in\bigcap_{\theta\in B}\operatorname{Ch}_T(\theta).
\]
\end{definition}

\begin{definition}[Joint choric mask]
A set $B\subseteq\Theta$ is \emph{jointly} $(\F_T,\eps)$-chorically masked at horizon $T$ if there is a single family $\{P_\theta\in\operatorname{Ch}_T(\theta)\}_{\theta\in B}$ with $\Delta_{\F_T}(P_\theta,P_{\theta'})\leq\eps$ for every pair $\theta,\theta'\in B$. Pairwise masking lets the witnesses vary with the pair; joint masking does not.
\end{definition}

\begin{theorem}[Balanced chorus]
A set $B\subseteq\Theta$ is jointly perfectly chorically masked for the full audience at horizon $T$ if and only if it admits a common public kernel.
\end{theorem}

\begin{proof}
For the full test class, $\Delta(P,Q)=0$ if and only if $P=Q$. A joint family with all pairwise distances zero therefore consists of one law $W$ repeated, and $W$ lies in $\operatorname{Ch}_T(\theta)$ for every $\theta\in B$: a common public kernel. Conversely, given a common public kernel $W$, the constant family $P_\theta=W$ is a joint witness.\end{proof}

\begin{remark}
Pairwise masking is weaker than group masking. Three convex sets may intersect pairwise while having empty triple intersection. A room that protects every pair separately may still fail to supply one public story that protects the whole group.
\end{remark}

\begin{definition}[Helly number of a room]
For fixed horizon $T$, the Helly number of the family $\{\operatorname{Ch}_T(\theta):\theta\in\Theta\}$ is the smallest integer $h$ such that every subfamily with empty intersection has a subfamily of size at most $h$ with empty intersection.
\end{definition}

Since all choric hulls live in the simplex over $Y^T$, whose affine dimension is at most $|Y|^T-1$, Helly's theorem gives a finite upper bound.

\begin{proposition}[Finite obstruction size]
If a set $B\subseteq\Theta$ has no common public kernel, then some subset $B_0\subseteq B$ of size at most $|Y|^T$ has no common public kernel.
\end{proposition}

\begin{proof}
The choric hulls are convex subsets of $\R^{|Y|^T}$ contained in the probability simplex, whose affine dimension is $|Y|^T-1$. By Helly's theorem~\cite{helly1923}, if the full intersection is empty, a subfamily of size at most $|Y|^T$ has empty intersection.\end{proof}

This proposition is a compactness statement for a finite public room. Failure of group masking has a finite certificate of bounded size. One does not need to inspect every hidden locus at once.

\section{Audience classes}

The audience is not a monolith. We isolate several test classes that occur naturally in public systems.

\begin{definition}[Count audience]
Given a subset $B\subseteq Y$, the count statistic is
\[
    n_B(y_{1:T})=\sum_{t=1}^T 1_{\{y_t\in B\}}.
\]
A count audience is any class of tests factoring through a finite list of count statistics.
\end{definition}

\begin{definition}[Phase audience]
Let $\mathcal P=\{I_1,\ldots,I_m\}$ be a partition of $\{1,\ldots,T\}$ into phases. A phase audience sees, for each phase $I_j$, the multiset of symbols appearing in that phase, but not their internal order.
\end{definition}

\begin{definition}[Ritual audience]
A ritual audience is generated by a finite family of forbidden or suspicious patterns
\[
    \mathcal R=\{R_1,\ldots,R_k\},\qquad R_i\subseteq Y^T.
\]
Its tests are bounded functions of the vector
\[
    (1_{R_1},\ldots,1_{R_k}).
\]
\end{definition}

\begin{definition}[Gossip closure]
Let $\F_T$ be a test class. Its gossip closure $\operatorname{Gos}(\F_T)$ is the smallest class containing $\F_T$ and closed under complements, finite maxima, finite minima, and thresholding of averages of finitely many tests.
\end{definition}

The gossip closure models weak social computation. A crowd may combine simple impressions, repeat them, and threshold them. It may not compute an optimal likelihood ratio, but it can turn small local tests into shared classifications.

\begin{proposition}[Closure monotonicity]
For every $P,Q$,
\[
    \Delta_{\F_T}(P,Q)
    \leq
    \Delta_{\operatorname{Gos}(\F_T)}(P,Q)
    \leq
    \TV(P,Q).
\]
\end{proposition}

\begin{proof}
The first inequality follows from inclusion. The second follows because every gossip-closed test is still a bounded test on $Y^T$, and total variation is the supremum over all bounded tests.\end{proof}

\section{Design rules extracted from the calculus}

The formal statements above produce concrete design rules for staged public systems. These rules are not extra assumptions. They are translations of the theorems.

\paragraph{Rule 1: remove mandatory unique gestures.}
If a hidden locus must produce a gesture that others cannot produce, the unique-gesture lower bound gives a positive exposure floor. The correct repair is not merely to add randomness. The room must either remove the mandatory gesture or give other loci enough cover mass on the same gesture.

\paragraph{Rule 2: audit frequencies, not only events.}
An at-least-once event may saturate under repetition. Frequency tests often remain sensitive. A room that looks safe under one-window event tests may fail under repeated count audiences.

\paragraph{Rule 2.5: separate trace exposure from attention exposure.}
A residue in \(Y_{1:T}\) is not the same as a residue in \(\widetilde Y_{1:T}\). Designs should record the attention lens used by the audience: gaze budget, dwell window, blanked coordinates, delay, and candidate splitting. Otherwise an apparent success may be only an aperture failure.

\paragraph{Rule 3: make the public quotient explicit.}
If the room intends the audience to see only a quotient of traces, the quotient should be part of the release system. Otherwise an audience may use a finer statistic than the designer intended.

\paragraph{Rule 4: separate first-room residue from link residue.}
A downstream release gate can add exposure even when the upstream traces are masked. The linked-room theorem gives the accounting identity: upstream distance plus link discrepancy.

\paragraph{Rule 5: require common kernels for group protection.}
Pairwise compatibility is not enough. Group masking needs a common public kernel, or at least an approximate intersection of all relevant choric hulls.

\section{Computation of the room radius}

Because the trace space is finite, the exposure radius can be computed by finite-dimensional optimization. We state this in a form that exposes the linear structure.

List the admissible role-script trace laws for $\theta$ as
\[
    P_{\theta,1},\ldots,P_{\theta,m_\theta}\in\R^{Y^T}.
\]
A choric law is $P_\theta x$ where $P_\theta$ is the matrix with these laws as columns and $x\in\Delta_{m_\theta}$.

For the full audience, the pairwise choric distance is
\[
    d_T(\theta,\theta')=\frac12
    \min_{x\in\Delta_{m_\theta},\ z\in\Delta_{m_{\theta'}}}
    \left\|P_\theta x-P_{\theta'}z\right\|_1.
\]
This is a linear program. Introduce nonnegative variables $u_y,v_y$ such that
\[
    P_\theta x-P_{\theta'}z=u-v.
\]
Then minimize
\[
    \frac12\sum_{y\in Y^T}(u_y+v_y)
\]
subject to simplex constraints on $x,z$. For a finite event audience generated by events $G_1,\ldots,G_k$, replace each trace law by the vector of event masses
\[
    \Phi(P)=(P(G_1),\ldots,P(G_k))\in[0,1]^k
\]
and solve
\[
    \min_{x,z}\left\|\Phi(P_\theta x)-\Phi(P_{\theta'}z)\right\|_\infty
\]
if the audience may choose the strongest event coordinate.

\begin{proposition}[Linear program for event audiences]
Let $\F_T=\{1_{G_1},\ldots,1_{G_k}\}$ up to complements. Then
\[
    \inf_{P,Q}\Delta_{\F_T}(P,Q)
    =\min_{x,z}\max_{1\leq i\leq k}
    \left|\Phi_i(P_\theta x)-\Phi_i(P_{\theta'}z)\right|,
\]
with $x,z$ ranging over the appropriate simplices. This value is the optimum of a linear program after introducing a scalar variable $t$ and constraints
\[
    -t\leq \Phi_i(P_\theta x)-\Phi_i(P_{\theta'}z)\leq t,
    \qquad i=1,\ldots,k.
\]
\end{proposition}

\begin{proof}
For this audience, the distance is the maximum absolute difference over the listed event probabilities. The displayed constraints are exactly the standard epigraph form for minimizing a maximum absolute value.\end{proof}

This computational form is part of the model, not an implementation afterthought. It lets a designer audit whether a proposed public grammar actually supplies cover mass.

\section{A failure mode: polite singularity}

A subtle failure occurs when every individual action is polite, allowed, and common, but the sequence is singular. We call this polite singularity.

\begin{definition}[Polite singularity]
A law $P\in\operatorname{Ch}_T(\theta)$ has polite singularity against $\theta'$ relative to $\F_T$ if every one-step marginal of $P$ is matched by some law in $\operatorname{Ch}_T(\theta')$, but no law in $\operatorname{Ch}_T(\theta')$ matches $P$ within distance $\eps$ on $Y^T$.
\end{definition}

\begin{example}[Matched gestures, exposed order]
Let $Y=\{a,b\}$ and $T=2$. Suppose locus $\theta$ can produce only the law concentrated on $ab$, while locus $\theta'$ can produce only the law concentrated on $ba$. The one-step symbol counts are identical: both traces contain one $a$ and one $b$. A count audience cannot distinguish them. An order audience can distinguish them perfectly. Thus the same room is safe under one quotient and exposed under another.
\end{example}

\begin{proposition}[Order residue]
In the preceding example, if $\F_T$ contains the event $\{ab\}$, then the distance is $1$. If $\F_T$ is the count audience for symbols $a,b$, the distance is $0$.
\end{proposition}

\begin{proof}
For the event $\{ab\}$, the probabilities are $1$ and $0$. For the count audience, both laws push forward to the same count vector. Apply quotient masking.\end{proof}

This example is small but important. Public systems often audit local frequencies and miss order. The audience, however, may remember sequence.

\section{Parastatic release}

The term \emph{parastatic} is used here for a release that stands beside the main trace rather than inside it. Examples include a public timestamp, a moderator note, an institutional label, or a delayed confirmation. Such releases are dangerous because they look auxiliary while changing the audience's test class.

\begin{definition}[Parastatic channel]
A parastatic channel for horizon $T$ is a Markov kernel
\[
    B:Y^T\to \Delta(W)
\]
into a finite side alphabet $W$. The augmented trace is $(Y_{1:T},W)$.
\end{definition}

\begin{proposition}[Side-channel expansion]
Let $P,Q$ be laws on $Y^T$ and let $B,B'$ be parastatic channels. The augmented laws satisfy
\[
    \TV(PB,QB')\leq \TV(P,Q)+\sup_y\TV(B(y),B'(y)).
\]
If $B=B'$, then
\[
    \TV(PB,QB)\leq \TV(P,Q).
\]
\end{proposition}

\begin{proof}
This is the linked-room theorem with the second room equal to the side alphabet.\end{proof}

A common parastatic channel cannot increase full-audience distance. A locus-dependent parastatic channel can. In public systems, the side label is often where the room leaks: not the message, but the manner in which the message is certified, delayed, grouped, or excused.

\section{From protection to accountability}

Masking is not the same as unaccountability. A room may hide a protected locus while preserving public accountability through a different quotient. The formal distinction is simple.

\begin{definition}[Accountability statistic]
An accountability statistic is a map
\[
    \kappa:Y^T\to V
\]
whose value is intended to remain visible. A privacy quotient is a map
\[
    \nu:Y^T\to Z
\]
whose value is intended to carry no protected-locus residue.
\end{definition}

A room can be designed so that $\kappa$ is informative while $\nu$ is masked. For example, $\kappa$ may record whether a required response occurred, while $\nu$ coarsens timing and role markers. The mathematical requirement is that the audience tests for protected inference factor through $\nu$, while accountability tests factor through $\kappa$.

\begin{definition}[Split visibility]
A room has $(\kappa,\nu)$-split visibility over $B\subseteq\Theta$ if there is a common public kernel for the $\nu$-pushforwards of all loci in $B$, while the distribution of $\kappa$ satisfies a specified accountability constraint $\mathcal L\subseteq \Delta(V)$.
\end{definition}

\begin{proposition}[Compatibility condition]
Split visibility is possible if and only if there exist choric laws $P_\theta\in\operatorname{Ch}_T(\theta)$ for each $\theta\in B$ such that
\[
    \nu_\#P_\theta=W\quad\text{for a common }W,
\]
and
\[
    \kappa_\#P_\theta\in\mathcal L
\]
for every $\theta\in B$.
\end{proposition}

\begin{proof}
This restates the definitions in pushforward form. Necessity follows by taking the laws that witness split visibility. Sufficiency follows because those laws provide the common privacy quotient and satisfy the accountability constraint.\end{proof}

This proposition is intentionally plain. It prevents a common confusion: hiding a protected coordinate does not require destroying every public statistic. It requires choosing which statistic must remain visible and which statistic must be neutralized.

\section{Boundary cases}

We record three boundary cases because they are common sources of drift.

\subsection{Randomness without cover}

Randomness helps only if it moves a trace law toward the relevant choric hull. If every admissible randomized script for $\theta$ keeps a gesture probability above $a$, and every admissible randomized script for $\theta'$ keeps it below $b<a$, no amount of local randomization closes the gap. This is exactly the unique-gesture lower bound.

\subsection{Coarsening without commitment}

A room may announce that the audience should ignore timing or order. That announcement has no mathematical force unless the release process actually coarsens the public trace or the audience test class is constrained. If exact traces remain visible, the full test class remains available.

\subsection{Pairwise comfort}

A room may pass every pairwise comparison and still lack a common group kernel. This is the Helly obstruction. Group protection is an intersection property, not a list of pairwise reassurances.

\subsection{Full-trace comfort}

A room may also confuse full-trace analysis with perceived-trace analysis. If the audience is bounded, the relevant law is \(LP\), not \(P\). Conversely, if the full trace remains available to a later or stronger audience, attention-filtered safety should not be reported as trace-level safety.

\section{Discussion}

The calculus developed here treats public behavior as a trace object. This is a narrow choice, but it is the source of the model's strength. The hidden locus is never directly inspected. The room is judged by what its trace allows an audience to infer. This makes the model suitable for public release systems where formal secrecy, social interpretation, and institutional procedure overlap.

The vocabulary is operationally staged: locus, role, gesture, chorus, residue, gate, and parastatic channel are bookkeeping devices for separating hidden coordinates, visible actions, admissibility constraints, and observer capabilities. The terms are useful only to the extent that they produce finite certificates: a separating test, a coupling witness, a hull intersection, a pressure ledger, or a localization failure.

Read institutionally, the debt theorems are an audit instrument. An institution under unresolved pressure widens selection; the externalization trichotomy says that widening either clears, localizes, or externalizes. An audit that records only breadth and cumulative capture cannot distinguish the three branches; an audit that records the localization kernel and the debt ledger can. This is the formal content behind standing critiques of categorical screening~\cite{gandy2006,tufekci2014,barocas2019,kearns2019}: the objection is not to response as such, but to response that cannot exhibit a localization certificate.

There are several directions in which the model can be sharpened.
\begin{enumerate}[label=(\roman*),leftmargin=2.2em]
\item Replace finite horizons by stopped processes and study exposure under random exit times.
\item Let the attention lens adapt after observing public residues over many releases, so that gaze itself becomes a strategic public resource.
\item Study choric debt under partially observed control, where unresolved pressure changes selection breadth before it improves localization.
\item Let the audience test class adapt after observing public residues over many releases.
\item Introduce costs for scripts, so that masking cannot be achieved by forcing one locus into an unreasonable role.
\item Study approximate common kernels with fairness constraints on role burden.
\item Couple the room to a dynamic norm gate that changes after public attention identifies a residue.
\end{enumerate}
The present paper keeps these extensions outside the formal core. The finite chamber already contains the main phenomenon: exposure is a property of trace, audience, and admissible cover.

\section{Conclusion}

An ambient release system is safe only when its visible grammar supplies enough cover for its hidden loci and when its audience model states what is actually perceived. The object is not a secret in isolation, but a trace inside a room, and often a filtered trace inside a bounded observer. Choric masking formalizes this idea by comparing convex hulls of admissible trace laws through audience tests and attention lenses. The resulting calculus identifies exact exposure carriers, proves lower bounds for unique gestures, separates trace-level residue from attention-level residue, shows how small residues amplify under repetition, proves how unresolved residue can become choric risk debt, and gives accounting rules for linked public releases. The theory is finite, inspectable, and compositional. It turns public readability into a finite mathematical object: every masking assertion has a finite witness, every exposure assertion has a public carrier, and a room's claim about itself is therefore always an exhibit.

\appendix
\section{Layer-cake residue for bounded tests}

We include the standard finite argument because it is useful for audits.

\begin{lemma}[Threshold residue]
Let $P,Q$ be laws on finite $\Omega$ and let $f:\Omega\to[0,1]$. If
\[
    \mathbb E_P f-\mathbb E_Q f>\eps,
\]
then there exists $t\in[0,1]$ such that
\[
    P\{f\geq t\}-Q\{f\geq t\}>\eps.
\]
\end{lemma}

\begin{proof}
For each $\omega$, $f(\omega)=\int_0^1 1_{\{f(\omega)\geq t\}}\,dt$. Therefore
\[
    \mathbb E_P f-\mathbb E_Q f
    =\int_0^1 \left(P\{f\geq t\}-Q\{f\geq t\}\right)dt.
\]
If every integrand were at most $\eps$, the integral would be at most $\eps$. Hence some threshold has residue greater than $\eps$.\end{proof}

\section{Approximate intersections}

For group masking, exact common kernels may be too strict. The finite geometry still gives a clean formulation.

\begin{definition}[Approximate common kernel]
For $B\subseteq\Theta$, a law $W$ on $Y^T$ is an $(\F_T,\eps)$-common kernel if for every $\theta\in B$ there exists $P_\theta\in\operatorname{Ch}_T(\theta)$ such that
\[
    \Delta_{\F_T}(P_\theta,W)\leq\eps.
\]
\end{definition}

\begin{proposition}[Approximate group masking]
If $B$ admits an $(\F_T,\eps)$-common kernel, then every pair $\theta,\theta'\in B$ is $(\F_T,2\eps)$-chorically masked.
\end{proposition}

\begin{proof}
Choose $P_\theta$ and $P_{\theta'}$ within $\eps$ of $W$. The triangle inequality gives
\[
    \Delta_{\F_T}(P_\theta,P_{\theta'})
    \leq \Delta_{\F_T}(P_\theta,W)+\Delta_{\F_T}(W,P_{\theta'})
    \leq 2\eps.
\]
\end{proof}

\section{A note on finite scripts}

The set of all stochastic scripts is infinite because probabilities vary continuously. For computation, one may either optimize directly over stochastic kernels or restrict to deterministic scripts and then take convex hulls. In finite-horizon controlled processes with perfect recall, randomized scripts induce trace laws that are mixtures over deterministic scripts.

\begin{proposition}[Deterministic expansion]
Fix $\theta,r,T$. Every trace law induced by a randomized admissible script is a convex combination of trace laws induced by deterministic admissible scripts for the same role.
\end{proposition}

\begin{proof}
A deterministic script chooses one admissible action at each visible history. A randomized script assigns a distribution over admissible actions at each visible history. Because the horizon and history tree are finite, the product of these local action distributions defines a distribution over deterministic scripts: sample in advance one action for every history according to the randomized kernel at that history. Conditional on the sampled deterministic script, the induced process follows that script. Averaging over the sampled deterministic scripts gives the same trace law as the original randomized script.\end{proof}

Thus $\operatorname{Ch}_T(\theta)$ may be generated from deterministic scripts alone. This converts the room radius into a finite linear optimization problem, although the number of deterministic scripts can be large.

\section{Notation}

The main notation is collected here in prose. The finite hidden loci are $\Theta$, visible roles are $R$, staged states are $S$, actions are $A$, observations are $Y$, and the norm gate is $N$. For a role $r$, the admissible scripts over horizon $T$ are $\Pi_T(r)$, and $P_{\theta,r}^{\pi,T}$ is the trace law induced by locus $\theta$, role $r$, and script $\pi$. The admissible trace set of a locus is $\M_T(\theta)$, its convex choric hull is $\operatorname{Ch}_T(\theta)$, the audience class is $\F_T$, and the corresponding distance is $\Delta_{\F_T}$. The signed residue of an event is $\res_{P,Q}(C)$. The exposure radius of a locus is $\mathfrak r_T$, and the room radius is $\mathfrak R_T$. In the stratified model, $\Lambda$ is the hidden carrier, $C_T^{(j)}(a)$ is the layer choric hull for coordinate value $a$, $\Phi_{\G_T}$ is the event measurement map, $H_T^{(j)}(a;\G_T)$ is the measurement hull, and $\operatorname{Loc}_{j,\eps}^{\F_T}$ is the compatible layer-value set. For attention, $L$ is a stochastic attention lens and $\mathfrak R_T^L$ is the attention-filtered room radius. For post-release dynamics, $X$ is the object-state space, $\Psi(x)$ is residual object pressure, $\Pi(x)$ is the downstream script set compatible with state $x$, $L_X$ is the burden link, $\mathsf B(\nu)$ is expected burden, and $\mathfrak E_L(C;x)$ is the attention-pressure certificate product. In repeated rooms, $D_k$ is choric risk debt, $u_k$ and $c_k$ are unresolved residue and localized clearance, $m(D_k)$ is debt-responsive selection breadth, and $h_{i,k}$ is the resolution hazard of candidate $i$.

\end{document}